\documentclass[a4paper,11pt]{article}

\usepackage{a4wide}
\usepackage[utf8]{inputenc}
\usepackage{soul} 
\usepackage[english]{babel} 
\usepackage{amsmath}
\usepackage{amsfonts}
\usepackage{amssymb}
\usepackage{amsthm}
\usepackage{mathrsfs}
\usepackage{url}
\usepackage{enumerate}
\usepackage{fixmath}
\usepackage{graphicx}
\usepackage[all]{xy}
\usepackage{algorithm}
\usepackage{algorithmic}
\usepackage{color}
\usepackage{hyperref}
\usepackage{booktabs}
\usepackage{marvosym}

\usepackage[textsize=small]{todonotes}

\usepackage{enumitem}
\usepackage{multirow}
\usepackage{pifont}

\addto\captionsenglish{}

\SetEnumitemKey{ncases}
  {itemsep=0pt,align=left,leftmargin=\parindent,
   itemindent=!,label={\normalfont\textbf{Case $\arabic*$}}}

\SetEnumitemKey{cases}
  {itemsep=0pt,align=left,leftmargin=\parindent,
   itemindent=!,before=}

\newtheorem{theorem}{Theorem}
\newtheorem{proposition}[theorem]{Proposition}
\newtheorem{corollary}[theorem]{Corollary}
\newtheorem{lemma}[theorem]{Lemma}
\newtheorem{notation}[theorem]{Notation}
\newtheorem{assumption}[theorem]{Assumption}

\theoremstyle{conjecture}

\theoremstyle{definition}
\newtheorem{definition}{Definition}

\theoremstyle{remark}
\newtheorem{remark}{Remark}

\newcommand{\C}[3]{\mathcal C_L(\mathcal{#1},\mathcal {#2},{#3})}

\renewcommand{\P}{\mathcal{P}}
\newcommand{\eqdef}{\stackrel{\text{def}}{=}}

\newcommand{\F}{\ensuremath{\mathbb{F}}}
\newcommand{\R}{\ensuremath{\mathbb{R}}}

\newcommand{\X}{\mathcal{X}}

\newcommand{\Ec}{\mathcal{E}}

\newcommand{\Lc}{\mathcal{L}}

\newcommand{\Xc}{\mathcal{X}}
\newcommand{\Yc}{\mathcal{Y}}
\newcommand{\Zc}{\mathcal{Z}}

\newcommand{\prob}{\ensuremath{\textsf{prob}}}
\newcommand{\code}[1]{\ensuremath{\mathscr{#1}}}

\newcommand{\word}[1]{\mathbf{#1}}

\newcommand{\cv}{\word{c}}

\newcommand{\pv}{\word{p}}
\newcommand{\piv}{{\mathbold{\pi}}}

\newcommand{\uv}{\word{u}}
\newcommand{\vv}{\word{v}}
\newcommand{\wv}{\word{w}}
\newcommand{\xv}{\word{x}}
\newcommand{\yv}{\word{y}}

\newcommand{\mat}[1]{\ensuremath{\boldsymbol{#1}}}
\newcommand{\Am}{\mat{A}}
\newcommand{\Bm}{\mat{B}}

\newcommand{\Gm}{\mat{G}}

\newcommand{\Jm}{\mat{J}}

\newcommand{\Mm}{\mat{M}}
\newcommand{\Mv}{\mat{M}}

\newcommand{\Pim}{\mat{\Pi}}

\newcommand{\zero}{\mat{0}}
\newcommand{\one}{\mat{1}}

\newcommand{\fq}{\F_{q}}
\newcommand{\Fq}{\F_{q}}
\newcommand{\Fqm}{\F_{q^m}}

\newcommand{\UVW}{\left(U\mid U+V \mid U+V+W\right)}

\renewcommand{\P}{\mathcal{P}}
\newcommand{\UV}{\left(U\mid U+V\right)}

\newcommand{\uuv}{\left(\mathbf u \mid \mathbf u+\mathbf v\right)}

\newcommand{\Oq}{\mathcal O\left(\frac{1}{q}\right)}

\newcommand{\OO}[1]{\mathcal{O}\left( #1 \right)}
\newcommand{\Th}[1]{\Theta\left( #1 \right)}

\newcommand{\esp}{{\mathbb{E}}}

\newcommand{\norm}[1]{\left|\!\left|#1\right|\!\right|}

\newcommand{\eqs}{\stackrel{\sigma}{=}}

\newcommand{\ire}[1]{{\textcolor{violet}{#1}}}

\newcommand{\Bha}[1]{{\mathcal{Z}\left( #1 \right)}}
\newcommand{\CKV}{C_{\text{KV}}}
\newcommand{\qSCp}{{q\text{-SC}_p}}
\newcommand{\Wc}[2]{{W(#2|#1)}}
\DeclareMathOperator{\trace}{trace}
\DeclareMathOperator{\diag}{diag}

\usepackage{authblk}

\begin{document}

\title{Attaining Capacity with  Algebraic Geometry Codes  through the $(U|U+V)$ Construction and Koetter-Vardy Soft Decoding}

\author{Irene M{\'a}rquez-Corbella\footnote{Dept. Mathematics, Statistic and Operation Research, University of La Laguna,  Email: \texttt{irene.marquez.corbella@ull.es}.},\quad Jean-Pierre Tillich\footnote{Inria, Email: \texttt{jean-pierre.tillich@inria.fr}.}}

\maketitle
\begin{abstract}
In this paper we show how to attain the capacity of discrete symmetric channels with polynomial time decoding complexity by considering iterated $\UV$ constructions with 
Reed-Solomon code or algebraic geometry code components. These  codes are decoded with a recursive computation of the 
{\em a posteriori} probabilities of the code symbols together with the Koetter-Vardy soft decoder used 
for decoding the code components in polynomial time. We show that when the
number of levels of the iterated $\UV$ construction tends to infinity, we attain the capacity of any discrete symmetric channel in this way.
This result follows from the polarization theorem together with a simple lemma 
explaining how the Koetter-Vardy decoder behaves for Reed-Solomon codes of rate close to $1$. However, even if this way of attaining 
the capacity of a symmetric channel is essentially the Ar{\i}kan polarization theorem, there are some differences with standard polar codes.
 Indeed, with this strategy we can operate succesfully close to channel
capacity even with a small number of levels of the iterated $\UV$ construction and the probability of error
decays quasi-exponentially with the codelength in such a case (i.e. exponentially if we forget about the 
logarithmic terms in the exponent). We can even improve on this result by considering the algebraic geometry codes constructed in 
\cite{TVZ82}. In such a case, the probability of error decays exponentially in the codelength for any rate below the capacity of the channel.
Moreover, when comparing this strategy to Reed-Solomon codes (or more generally algebraic geometry codes) decoded with the
Koetter-Vardy decoding algorithm, it does not only improve the noise level that the code can tolerate, it also results in a significant complexity gain.
 \end{abstract}

\section{Introduction}\label{sec:introduction}
 \paragraph{Improving upon the error correction performance of Reed-Solomon codes.}
 Reed-Solomon codes are among the most extensively used error correcting codes. 
It has long been known how to decode them
up to half the minimum distance. This gives a decoding algorithm that is able to 
correct a fraction $\frac{1-R}{2}$ of errors in a Reed-Solomon code of  rate $R$. 
However, it is only in the late nineties that a breakthrough was obtained
in this setting with Sudan's algorithm \cite{S97b} and its improvement in \cite{GS99} who showed how to go beyond this barrier with an algorithm 
which in its \cite{GS99} version decodes any fraction of errors smaller than $1 - \sqrt{R}$. 
 This exceeds the minimum distance
bound $\frac{1-R}{2}$ in the whole region of rates $[0,1)$. Later on, it was shown that this decoding algorithm could also be modified a little bit in order to cope with soft information on the errors \cite{KV03}. A few years later, it was also realized by Parvaresh and Vardy in \cite{PV05} 
that by a slight modification of Reed-Solomon codes and by an increase of the alphabet size it was possible to beat the $1-\sqrt{R}$ decoding radius. 
Their new family of codes is  list decodable beyond this radius for low rate. 
Then, Guruswami and Rudra \cite{GR06} improved on these codes by presenting a new family of codes, namely {\em folded Reed-Solomon codes} 
with a polynomial time decoding algorithm achieving the list decoding capacity $1-R-\epsilon$ for every rate $R$ and $\epsilon >0$.

The initial motivation of this paper is to present another modification of Reed-Solomon codes that improves the fraction of errors
that can be corrected.  It consists in using them in a $\UV$ construction. In other words, we choose in this construction $U$ and $V$
to be Reed-Solomon codes. 
 We will show that, in the low rate regime, this class of codes outperforms a little bit a 
 Reed-Solomon code decoded with the Guruswami and Sudan decoder. The point is that this $\UV$ code can be decoded in two steps : 
\begin{enumerate}
\item First by subtracting the left part $y_1$ to the right part $y_2$ of the received vector $(y_1|y_2)$ 
and decoding it with respect to $V$. In such a case, we are
left with decoding a Reed-Solomon code with about twice as many errors. 
\item Secondly, once we have 
recovered the right part $v$ of the codeword, we can get a word $(y_1,y_2-v)$ which should
match two copies of a same word $u$ of $U$. 
We can model this decoding problem by having some
soft information on the received word when we have sent $u$. 
\end{enumerate}
It turns that this channel error model is much less noisy than the original $q$-ary symmetric channel we started with. This soft information can be used in Koetter and Vardy's decoding algorithm.
By this means we can choose $U$ to be a Reed-Solomon code of much bigger rate than $V$. All in all, it turns out that by choosing $U$ and $V$ with appropriate rates we can beat the $1 - \sqrt{R}$ bound of Reed-Solomon codes in the low-rate regime.

It should be noted however that beating this $1- \sqrt{R}$ bound comes at the cost of having now an algorithm which does not work as for the aforementioned
papers \cite{S97b,GS99,PV05,GR06} for every error of a given weight (the so called adversarial error model) but with probability $1-o(1)$ for errors of a given weight. However contrarily to  \cite{PV05,GR06} which results in a significant increase of the alphabet size of the code, our alphabet size actually 
decreases when compared to a Reed-Solomon code: it can be half of the code length and can be even smaller when we apply this construction recursively. Indeed, we will show
that we can even improve the error correction performances by applying this construction again to the $U$ and $V$ components, i.e  
 we can choose $U$ to be a $(U_1|U_1+V_1)$ code
 and we replace in the same way the 
Reed-Solomon code $V$ by a $(U_2|U_2+V_2)$ code where $U_1,U_2,V_1$ and $V_2$ are  Reed-Solomon codes
(we will say that these $U_i$'s and $V_i$'s codes are the consituent codes of the iterated $\UV$-construction).
This improves slightly the decoding performances again in the low rate regime.

\paragraph{Attaining the capacity by letting the depth of the construction go to infinity with an exponential decay of the probability of error after decoding.}
The first question raised by these results is to understand what happens when we apply this iterative construction a number of times
which goes to infinity with the codelength. In this case, the channels faced by  the constituent Reed-Solomon codes polarize: they become either very noisy channels or very
clean channels of capacity close to $1$. This is precisely the polarization phenomenon discovered by Ar{\i}kan in \cite{A09}. Indeed this iterated $\UV$-construction is nothing but a standard polar code when the constituent codes are Reed-Solomon codes of length $1$ (i.e. just a single symbol).
The polarization phenomenon together with a result proving that the Koetter-Vardy decoder is able to operate sucessfully at rates close to $1$
for channels of capacity close to $1$ can be used to show that it is possible to choose the rates of the constituent Reed-Solomon codes in such a way 
that the code construction together with the Koetter-Vardy decoder is able to attain the capacity of symmetric channels.
On a theoretical level, proceeding in this way would not change however the asymptotics of the decay of the probability of error after decoding: the codes obtained in this way would still behave
as polar codes and would in particular have a probability of error which decays exponentially with respect to (essentially) the square root of the 
codelength. 

The situation changes completely however when we allow ourself to change the input alphabet of the channel and/or to use Algebraic Geometry (AG) codes.
The first point can be achieved by grouping together the symbols and view them as a symbol of a larger alphabet. The second point is also relevant here
since the Koetter and Vardy decoder also applies to AG codes (see \cite{KV03a}) with only a rather mild penalty in the error-correction capacity
related to the genus of the curve used for constructing the code. Both approaches can be used to overcome the limitation of having 
constituent codes in the iterated $\UV$-construction whose length is upper-bounded by the alphabet size. When we are allowed to choose long enough 
constituent codes the asymptotic behavior changes radically. We will indeed show that if we insist on using Reed-Solomon codes in the code
construction we obtain a quasi-exponential decay of the probability of error in terms of the codelength (i.e. exponential if we 
forget about the logarithmc terms in the exponent) and an exponential decay if we use the right AG codes. This improves very significantly upon polar codes.
Not only are we able to attain the channel capacity with a polynomial time decoding algorithm with this approach but we are also able to do 
so with an exponential decay of the probability of error after decoding. In essence, this sharp decay of the probability of error after decoding is due
to a result of this paper (see Theorems \ref{th:exponential} and \ref{th:exponentialAG}) showing that even if the Koetter-Vardy decoder is not able to attain the capacity with a probability of error going to zero
as the codelength goes to infinity  its probability of error decays like $2^{-K \epsilon^2 n}$ where 
$n$ is the codelength and $\epsilon$ is the difference between a quantity which is strictly smaller than the capacity of the channel
and the code-rate.

{\bf Notation.}
Throughout the paper we will use the following notation.
\begin{itemize}
\item  A linear code of length $n$, dimension $k$ and distance $d$ over a finite field $\mathbb F_q$ is referred to as an $[n,k,d]_q$-code. 
\item The concatenation of two vectors $\xv$ and $\yv$ is denoted by $(\xv|\yv)$. 
\item For a vector $\xv$ we either denote by $x(i)$ or by $x_i$ the $i$-th coordinate of $\xv$.
We use the first notation when the subscript is already used for other purposes or when there is 
already a superscript for $\xv$.
\item For a vector $\xv=(x_\alpha)_{\alpha \in \Fq}$ we denote by $\xv^{+\beta}$ the vector
$(x_{\alpha+\beta})_{\alpha \in \Fq}$.
\item For a matrix $\Mm$ we denote by $\Mm^j$ the $j$-th column of $\Mm$.
\item By some abuse of terminology, we also view a discrete memoryless channel $W$ with input alphabet $\Xc$ and 
output alphabet $\Yc$ as an $\Xc \times \Yc$ matrix whose $(x,y)$ entry is denoted by 
$W(y|x)$ which is defined as the probability of receiving $y$ given that $x$ was sent. 
We will identify the channel with this matrix later on.
\end{itemize}

\section{The code construction and the link with polar codes}
\label{sec:polar}

\noindent
{\bf Iterated $\UV$ codes.} This section details the code construction we deal with. It can be seen as a variation of polar codes 
and is nothing but an iterated $\UV$ code construction.
We first recall the definition of a $\UV$ code.   We refer to \cite[Th.33]{MS86} for the statements on the dimension and 
minimum distance that are given below.

\begin{definition}[$\UV$ code]
Let $U$ and $V$ be two codes of the same length and defined over the same finite field $\Fq$.
We define the $\left(U\mid U+V\right)$-construction of $U$ and $V$ as the linear code:
$$\UV = \left\{ (\mathbf u\mid\mathbf u + \mathbf v); \mathbf u \in U \hbox{ and } \mathbf v \in V\right\}.$$
The dimension of the $\UV$ code is $k_U+k_V$ and its minimum distance is $\min(2d_U,d_V)$ when the dimensions of $U$ and 
$V$ are $k_U$ and $k_V$ respectively, the minimum distance of $U$ is $d_U$ 
and the minimum distance of $V$ is $d_V$.
\end{definition}

The codes we are going to consider here are iterated $\UV$ constructions defined by
\begin{definition}[iterated $\UV$-construction of depth $\ell$]
\label{def:iterated_uv}
An iterated $\UV$-code $U_\epsilon$ of depth $\ell$ is defined from a set of $2^\ell$ codes $\left\{U_\xv;\xv \in \{0,1\}^\ell\right\}$
which have all the same length and are defined over the same finite field $\Fq$  by using the recursive definition
\begin{eqnarray*}
U_\epsilon & \eqdef & (U_0\mid U_0+U_1)\\
U_\xv &\eqdef &(U_{\xv\mid0}\mid U_{\xv\mid0}+ U_{\xv\mid1}) \;\;\text{for $\xv \in \{0,1\}^i$, $i \in \{1,\dots,\ell-1\}$}.
\end{eqnarray*}
The codes $U_\xv$ for $\xv \in \{0,1\}^\ell$ are called the {\em constituent codes} of the construction.
\end{definition}

In other words, an iterated $\UV$-code of depth $1$ is nothing but a standard $\UV$-code and an iterated $\UV$-code of depth $2$ is a $\UV$-code where
$U$ and $V$ are themselves $\UV$-codes.

\noindent{\bf Graphical representation of an iterated $\UV$ code.}
Iterated $\UV$-codes can be represented by complete binary trees in which each node has exactly two children except  the leaves.
A $\UV$-code is represented by a node with two childs, the  left child representing the $U$ code and the right child representing the $V$ code.
The simplest case is given is given in  Figure \ref{fig:uv}. Another example is given in Figure \ref{fig:example} and represents an iterated $\UV$-code $U_\epsilon$ of depth $3$ with a binary tree of depth $3$ whose leaves are the $8$ constituent codes of this construction.

\begin{figure}[h!]
\caption{Graphical representation of a $\UV$-code. \label{fig:uv}}
\centering
\includegraphics[width=4cm]{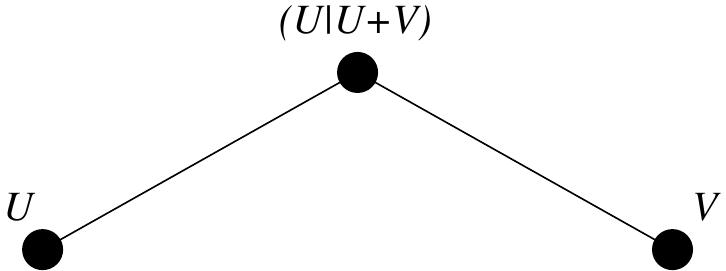}
\end{figure}

\begin{figure}[h!]
\caption{Example of an iterated $\UV$ code of depth $3$. \label{fig:example}}
\centering
\includegraphics[width=8cm]{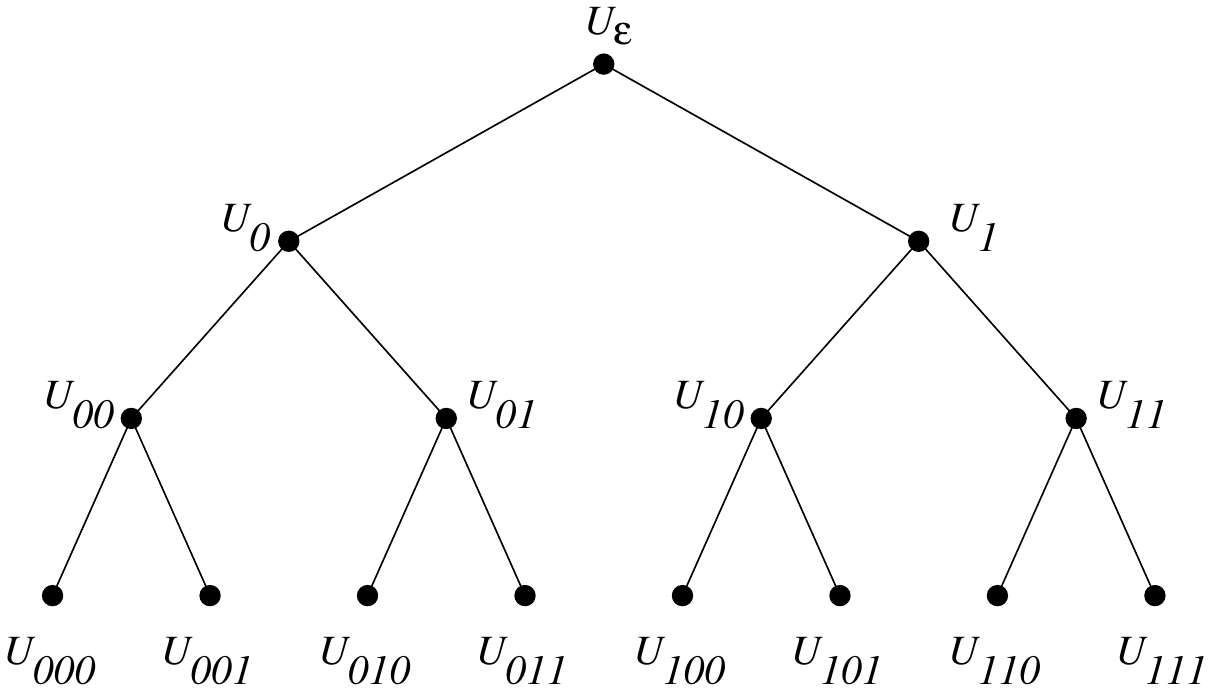}
\end{figure}

\begin{remark}
Standard polar codes (i.e. the ones that were constructed by Ar{\i}kan in \cite{A09}) are clearly a special case of the iterated $\UV$ construction.
Indeed such a polar code of length $2^\ell$  can be viewed as an iterated $\UV$-code of depth $\ell$ where the set $\left\{U_\xv;\xv \in \{0,1\}^\ell\right\}$
of constituent codes are just codes of length $1$. In other words, standard polar codes correspond to binary trees where all leaves are just single bits.
\end{remark}

\noindent{\bf Recursive soft decoding of an iterated $\UV$-code.}
As explained in the introduction our approach is to use the same decoding strategy as for Ar{\i}kan polar codes (that is his successive cancellation decoder) but by using 
now leaves that are codes which are much longer than single symbols.
This will have the effect of lowering rather significantly the error probability of error after decoding when compared
to standard polar codes. It will be helpful to change slightly the way the successive cancellation decoder is generally explained. Indeed this decoder can be viewed 
as an iterated decoder for a $\UV$-code, where decoding the $\UV$-code consists in first decoding the $V$ code and then the $U$ code with a decoder using soft information
in both cases. This decoder was actually considered before the invention of polar codes and has been considered for
decoding for instance Reed-Muller codes based on the fact that they are $\UV$ codes \cite{D06a,DS06}.

Let us recall how such a $\UV$-decoder works.
Suppose we transmit the codeword $\uuv\in \UV$ over a noisy channel and we receive the vector:
$\yv = (\mathbf y_1 \mid \mathbf y_2) $. We denote by $p(b\mid a)$ the probability of receiving $b$ when $a$ was sent
and assume a memoryless channel here. We also assume that all the codeword symbols $u(i)$ and  $v(i)$ are uniformly distributed.

\begin{itemize}
\item[\bf{Step 1.}] We first decode  $V$. We compute  the probabilities $\prob(v(i)=\alpha|y_1(i),y_2(i))$ for all positions $i$ and all $\alpha$ in $\fq$. 
 Under the assumption that we use a memoryless channel and that the $u(i)$'s and the $v(i)$'s are uniformly distributed for all $i$, it is straightforward to check that this probability is given by
 \begin{equation}
\label{eq:sum}
 \prob(v(i)=\alpha|y_1(i),y_2(i)) = \sum_{\beta \in \Fq} p(y_1(i)|\beta) p(y_2(i)|\alpha+\beta)
 \end{equation}
\item[\bf{Step 2.}] We use now Ar{\i}kan's successive decoding approach and assume that the $V$ decoder was correct and thus we have recovered $\vv$.
We compute now for all $\alpha \in \fq$ and all coordinates $i$ the probabilities $\prob(u(i)=\alpha|y_1(i),y_2(i),v(i))$ by using the
formula
\begin{equation}
\label{eq:product}
\prob(u(i)=\alpha|y_1(i),y_2(i),v(i)) = 
\frac{p(y_1(i) \mid \alpha) p(y_2(i) \mid \alpha + v(i))} {\sum_{\beta\in \mathbb F_q} p(y_1(i) \mid \beta) p(y_2(i) \mid \beta + v(i))}
\end{equation}
This can be considered as soft-information on $\uv$ which can be used by a soft information decoder for $U$.
\end{itemize} 

This decoder can then be used recursively for decoding an iterated $\UV$-code.
For instance if we denote by $U_\epsilon$ an iterated $\UV$-code of depth $2$ derived from the set of codes $\{U_{00},U_{01},U_{10},U_{11}\}$, the decoding works as follows (we used here the same notation as  in Definition \ref{def:iterated_uv}).
\begin{itemize}
\item {\bf Decoder for $U_1 = \left(U_{10} \mid U_{10}+U_{11}\right)$}. We first compute the probabilities for decoding $U_{11}$, this code is decoded with a soft information decoder. Once we have recovered the $U_{11}$ part (we denote the corresponding codeword by $\uv_{11}$), we can compute the relevant probabilities for decoding the $U_{10}$ code. This code is also decoded with a soft information decoder and we output a codeword $\uv_{10}$. All this work allows to recover the $U_1$ codeword denoted by $\uv_1$ by combining the 
$U_{10}$ and $U_{11}$ part as $\uv_1 = (\uv_{10}\mid \uv_{10}+\uv_{11})$.
\item {\bf Decoder for $U_0 = \left(U_{00} \mid U_{00}+U_{01}\right)$}. Once the $U_1$ codeword is recovered we can compute the probabilities for decoding the code $U_0$ and we decode this code in the same way as we decoded the code $U_1$.
\end{itemize}
Figure \ref{fig:order} gives the order in which we recover each codeword during the decoding process.

\begin{figure}[h!]
\caption{This figure summarizes in which order we recover each codeword of a $\UV$ code of depth $2$. 
Nodes in red represent codes that are decoded with a soft information decoder, nodes in black correspond to 
codes that are not decoded directly and whose decoding is accomplished by first recovering the two descendants of the node and then combining them
to recover the codeword we are looking for at this node.\label{fig:order}}
\centering
\includegraphics[width=8cm]{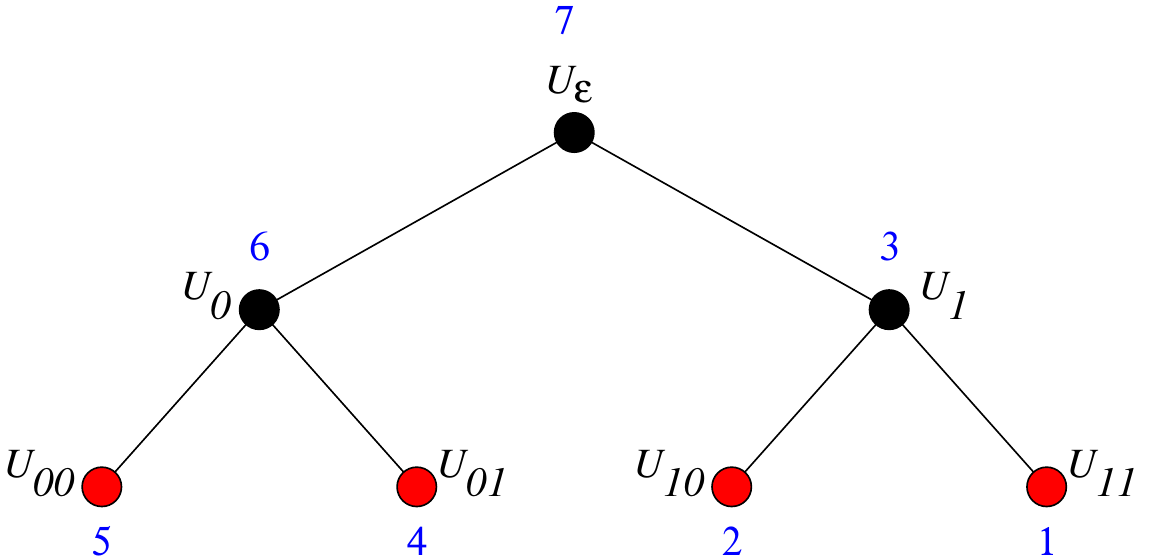}
\end{figure}

When the constituent codes of this recursive $\UV$ construction are just codes of length $1$, it is readily seen that this decoding simply amounts to the successive cancellation 
decoder of Ar{\i}kan. We will be interested in the case where these constituent codes are longer than this.
In such a case, we have to use as constituent codes, codes for which we have an efficient but possibly suboptimal decoder which can make use of soft information. Reed-Solomon
codes or algebraic geometry codes with the Koetter Vardy decoder are precisely codes with this kind of property. 

\noindent
{\bf Polarization.}
The probability computations made during the $\UV$ decoding \eqref{eq:sum} and \eqref{eq:product} correspond in a natural way to changing the channel model for the $U$ code and for
the $V$ code. These two channels really correspond to the two channel combining models considered for polar codes. More precisely, if we consider a memoryless channel of input alphabet $\Fq$ and output alphabet 
$\Yc$ defined by a transition matrix $W= (W(y|u))_{\substack{u \in \Fq \\ y \in \Fq}}$, then the channel viewed by the $U$ decoder, respectively the $V$ decoder is a memoryless channel with transition matrix $W^0$ and 
$W^1$ respectively, which are given by
\begin{eqnarray*}
 W^0(y_1,y_2,u_2|u_1) & \eqdef & \frac{1}{q} W(y_1|u_1) W(y_2|u_1 \oplus u_2) \\
 W^1(y_1,y_2|u_2) & \eqdef & \frac{1}{q} \sum_{u_1 \in \F_q} W(y_1|u_1) W(y_2|u_1 \oplus u_2)
 \end{eqnarray*}
Here the $y'_i$'s belong to $\Yc$ and the $u_i$'s belong to $\Fq$. If we define the channel $W^x$ for $x=(x_1\dots x_n) \in \{0,1\}^n$ recursively by 
$$
W^{x_1 \dots x_{n-1} x_n} = \left(W^{x_1 \dots x_{n-1} } \right)^{x_n}
$$
then the channel viewed by the decoder for one of the constituent codes $U_{x_1 \dots x_n}$  of an iterated $\UV$ code of depth $n$ 
(with the notation of Definition \ref{def:iterated_uv}) is nothing but the 
channel $W^{x_1 \dots x_n}$.

The key result used for showing that polar codes attain the capacity is that these channels polarize in the following sense
\begin{theorem} [{\cite[Theorem 1]{STA09} and \cite[Theorem 4.10]{S11g}}]\label{th:polarization}
Let $q$ be an arbitrary prime.  Then for a discrete $q$-ary input channel $W$ of symmetric capacity 
\footnote{Recall that the symmetric capacity of such a channel is defined as the mutual information between 
a uniform input and the corresponding output of the channel, that is 
$C \eqdef \frac{1}{q}\sum_{\alpha \in \Fq}\sum_{y \in \Yc} W(y|\alpha) \log_q \frac{W(y|\alpha)}{\sum_{\beta \in \Fq} \frac{1}{q} W(y|\beta)}$, 
where $\Yc$ denotes the output alphabet of the channel. 
} 
$C$ 
we have for all
$0 < \beta < \frac{1}{2}$
$$
\lim_{\ell \rightarrow \infty} \frac{1}{n} \left| i \in \{0,1\}^\ell : \Bha{W^i} \leq 2^{-n^\beta} \right| = C,
$$
where $n \eqdef 2^\ell$.
\end{theorem}
Here $\Bha{W}$ denotes the Bhattacharyya parameter of $W$ which is assumed to be a memoryless channel with
 $q$-ary inputs and outputs in an alphabet $\Yc$. It is given by
 \begin{equation}
\label{eq:Bha0}
 \Bha{W} \eqdef \frac{1}{q(q-1)} \sum_{x,x' \in \F_q, x' \neq x} \sum_{y \in \Yc} \sqrt{W(y|x) W(y|x')}
 \end{equation}
Recall that this Bhattacharrya parameter quantifies the amount of noise in the channel. It is close to $0$ for channels with very low noise (i.e. channels of capacity close
to $1$) whereas it is close to 
$1$ for very noisy channels (i.e. channels of capacity close to $0$).

\section{Soft decoding of Reed-Solomon codes with the Koetter-Vardy decoding algorithm}
\label{Section3}
It has been a long standing open problem  to obtain an efficient soft-decision decoding algorithm for Reed-Solomon
codes until Koetter and Vardy showed in \cite{KV03} how to modify appropriately the Guruswami-Sudan decoding algorithm in order to achieve this purpose. The complexity of this algorithm is polynomial and we will show here that the probability of error decreases exponentially  in the codelength when the noise level is below a certain threshold. 
Let us first review a few basic facts about this decoding algorithm.

\paragraph{\bf The reliability matrix.}
The Koetter-Vardy decoder \cite{KV03} is based on a  {\em reliability matrix} $\Pim_{\yv}$
of the codeword symbols $x(1),\dots,x(n)$ computed from the knowledge of the received word $\yv$ and which is defined by 
$$
\Pim_{\yv} = \left(\prob(x(j)=\alpha|y(j))\right)_{\substack{\alpha \in \fq \\ 1 \leq j \leq n}}
$$
Recall that the $j$-th column of this matrix $\Pim_{\yv}$ is denoted by $\Pim_{\yv}^j$. It  gives the a posteriori probabilities (APP) that the $j$-th codeword symbol is equal to $\alpha$ where $\alpha$ ranges over $\fq$.

We will be particularly interested in the $q$-ary symmetric channel model. 
The $q$-ary symmetric channel with error probability $p$, denoted by $q\hbox{-SC}_{p}$, takes a $q$-ary symbol at its input and outputs either the unchanged symbol, with probability $1-p$, or any of the other $q-1$ symbols, with probability $\tfrac{p}{q-1}$. Therefore, if the channel input symbols are uniformly distributed, the reliability matrix $\Pim_{\mathbf y}$ for $\qSCp$ is given by

$$\Pim_{\mathbf y}^j (\alpha) = \prob\left(x(j)=\alpha\mid y(j)\right) = \left\{ \begin{array}{ll}
1-p & \hbox{ if } \alpha = y(j)\\
\frac{p}{q-1} & \hbox{ if } \alpha \neq y(j)
\end{array}\right.$$

Thus, all columns of $\Pim_{\mathbf y}$ are identical up to permutation:
$$\Pim_{\mathbf y}^i =  \left(\begin{array}{c}
1-p \\ \frac{p}{q-1} \\ \vdots \\ \frac{p}{q-1}
\end{array}\right) \hbox{ (up to permutation)} $$
with $i=1, \ldots, n$.

This matrix is used by the Koetter-Vardy decoder
to compute a multiplicity matrix that serves as the input to its soft interpolation step.
When used in a $\UV$ construction and decoded as mentioned before, we will need to understand how the reliability matrix 
behaves through the $\UV$ decoding process. This is what we will do now.

\paragraph{\bf Reliability matrix for the $V$-decoder.} We denote the reliability matrix of the  $V$ decoder
by $(\Pim \oplus \Pim)_{\yv}$ 
when $\Pim_{\yv_1}$ and $\Pim_{\yv_2}$ are the initial reliability matrices corresponding to the two halves
of the received word $\yv=(\yv_1,\yv_2)$.
From the definition of the reliability matrix and \eqref{eq:sum} we readily obtain that 
\begin{equation}
\label{eq:pi_oplus}
(\Pim\oplus \Pim)_{\mathbf y}^i (\alpha) \eqdef \sum_{\beta \in \mathbb F_q} \Pim_{\mathbf y_1}^i (\beta) \cdot \Pim_{\mathbf y_2}^i (\alpha - \beta).
\end{equation}

\paragraph{\bf Reliability matrix for the $U$-decoder} 
Similarly, by using \eqref{eq:product} we see that the reliability matrix of the $U$ decoder, that we denote by
$\Pim \times \Pim_{\yv,\vv}$ is given by
\begin{equation}
\label{eq:pi_otimes}( \Pim\times \Pim)_{\mathbf y, \vv}^i(\alpha) = \frac{\Pim_{\yv_1}^i (\alpha) \cdot \Pim_{\mathbf \yv_2}^i (\alpha+v(i))}{\sum_{\beta \in \mathbb F_q} \Pim_{\mathbf \yv_1}^i  (\beta ) \cdot \Pim_{\mathbf \yv_2}^i(\beta+v(i))}.
\end{equation}
To simplify notation we will generally avoid the dependency on $\yv$ and $\vv$ and simply write $\Pim \oplus \Pim$ and $\Pim\times \Pim$.

\paragraph{\bf When does the Koetter-Vardy decoding algorithm succeed ?}
Let us recall how the Koetter-Vardy soft decoder \cite{KV03} can be analyzed.
 By \cite[Theorem 12]{KV03} their decoding algorithm outputs a list that contains the codeword $\mathbf c\in C$ if
\begin{equation}
\label{eq:success}
\frac{\left\langle \Pim, \lfloor \mathbf c\rfloor \right\rangle}{\sqrt{\left\langle \Pim, \Pim\right\rangle}} \geq \sqrt{k-1}+o(1)
\end{equation}
as the codelength $n$ tends to infinity,
where $\lfloor \mathbf c\rfloor$ represents a $q\times n$ matrix with entries $c_{i,\alpha} = 1$ if $c_i = \alpha$, and $0$ otherwise; and
$\left\langle \Am, \Bm \right\rangle$
 denotes the inner product of the two $q\times n$ matrices $A$ and $B$, i.e. 
$$\left\langle \Am, \Bm \right\rangle \eqdef \sum_{i=1}^q\sum_{j=1}^n a_{i,j}b_{i,j}.$$
The algorithm uses a parameter $s$ (the total number of interpolation points counted with multiplicity). The little-O $o(1)$ depends on 
the choice of this parameter and the parameters $n$ and $q$. 

We need a more precise formulation of the little-O of (\ref{eq:success}) to understand that we can get arbitrarily close to the lower bound $\sqrt{k-1}$ with polynomial complexity. 
In order to do so, let us provide more details about the Koetter Vardy decoding algorithm. Basically this algorithm starts by computing with Algorithm A of \cite[p.2814]{KV03} 
from the knowledge of the reliability matrix
$\Pim$ and for the aforementioned integer parameter $s$ a $q \times n$ nonnegative integer matrix $\Mv(s)$ whose entries sum up to $s$.
When $s$ goes to infinity
$\Mv(s)$ becomes proportional to $\Pim$. The cost of this matrix (we will drop the dependency in $s$) $C(\Mv)$ is defined as

\begin{equation}
\label{eq:cost}
C(\Mv) \eqdef \frac{1}{2} \sum_{i=1}^q \sum_{j=1}^n m_{ij}(m_{ij}+1) = \frac{1}{2} 
\left( \left\langle \Mv,  \Mv   \right\rangle + \left\langle \Mv, \one \right\rangle \right)
\end{equation}
where $m_{ij}$ denotes the entry of $\Mv$ at row $i$ and column $j$ and $\one$ is the
all-one matrix. The complexity of the Koetter-Vardy decoding algorithm is dominated by solving a system of $C(\Mv)$ linear equations.
Then, the number of codewords on the list produced by the Koetter-Vardy decoder for a given multiplicity matrix $\Mv$ does not exceed
$$
\Lc(\Mv) \eqdef \sqrt{\frac{2C(\Mv)}{k-1}}.
$$
It is straightforward to obtain from these considerations a soft-decision list decoder with a list which does not 
exceed some prescribed quantity $L$. Indeed it suffices to increase the value of $s$ in \cite[Algorithm A]{KV03}
 until getting a matrix $\Mv$ which is such that
$$
L \leq \Lc(\Mv) < L+1
$$
and to use this multiplicity matrix $\Mv$ in the Koetter-Vardy decoding algorithm. By following the terminology of \cite{KV03}
we refer to this decoding procedure as {\em algebraic soft-decoding with list size limited to $L$}.
\cite[Theorem 17]{KV03} explains that convergence to the $\sqrt{k-1}$ lower-bound is at least as fast as $\OO{1/L}$ 
\begin{theorem}[Theorem 17, \cite{KV03}]
\label{th:loose}
Algebraic soft-decoding with list size limited to $L$ produces a codeword $\cv$ if
\begin{equation}
\label{eq:loose}
\frac{\left\langle \Pim, \lfloor \mathbf c\rfloor \right\rangle}{\sqrt{\left\langle \Pim, \Pim\right\rangle}} \geq 
\frac{\sqrt{k-1}}{1 - \frac{1}{L}\left( \frac{1}{R^*}+\frac{\sqrt{q}}{2\sqrt{R^*}} \right)}
= \sqrt{k-1} \left( 1 + \OO{\tfrac{1}{L}} \right)
\end{equation}
where $R^* \eqdef \frac{k-1}{n}$ and the constant in $\OO{\cdot}$ depends only on $R^*$ and 
$q$.
\end{theorem}

\begin{remark}
\begin{enumerate}
\item
This theorem shows that the size of the list required to approach the asymptotic performance does not depend (directly) on the length
of the code, it may depend on the rate of the code and the cardinality of the alphabet though.
\item
As observed in \cite{KV03}, this theorem is a very loose bound. The actual performance of algebraic soft-decoding with list size limited to $L$ 
is usually orders of magnitude better than that predicted by \eqref{eq:loose}. A somewhat better bound is given by \cite[(44) p. 2819]{KV03} where the condition for successful decoding $\cv$ is
\begin{eqnarray}
\frac{\left\langle \Pim, \lfloor \mathbf c\rfloor \right\rangle}{\sqrt{\left\langle \Pim, \Pim\right\rangle}} 
& \geq &
\frac{\sqrt{k-1}}{1 - \frac{1}{L}\left( \frac{1}{R^*}+\frac{\sqrt{n}}{2\sqrt{R^* \left\langle \Pim, \Pim\right\rangle }} \right)}
 \nonumber \\
& \approx & 
\frac{\sqrt{k-1}}{1 - \frac{1}{L}\left( \frac{1}{R^*}+\frac{1}{2\sqrt{R^*}} \right)}.
\end{eqnarray}
where the approximation assumes that $\left\langle \Pim, \Pim\right\rangle \approx n$ which holds for noise levels of practical interest.
Note that this strengthens a little bit the constant in $\OO{\cdot}$ that appears in Theorem \ref{th:loose}, since it would  not depend on $q$ anymore.
\end{enumerate}
\end{remark}

\paragraph{\bf Decoding capability of the Koetter-Vardy decoder when the channel is symmetric.}
The previous formula does not explain directly under which condition on the rate of  the Reed-Solomon code
 decoding typically succeeds (in some sense this would be a ``capacity'' result for the Koetter-Vardy decoder).
 We will derive now such a result that appears to be new (but see the discussion at the end of this section). It will be convenient 
to restrict a little bit the class of memoryless channels we will consider- this will simplify formulas a great deal.
The idea  underlying this restriction is to make the behavior of the quantity 
$\left\langle \Pim, \lfloor \mathbf c\rfloor \right\rangle$ which appears in the condition of successful decoding \eqref{eq:success}
independent of the codeword $\cv$ which is sent. This is readily obtained by restricting the channel to be {\em weakly symmetric}.

\begin{definition}[weakly symmetric channel]
\label{def:weakly_symmetric}
A discrete memoryless $W$ with input alphabet $\Xc$ and output alphabet $\Yc$ 
is said to be weakly symmetric if and only if  there is a partition of 
the output alphabet $\Yc= Y_1 \cup \dots \cup Y_n$ such that all the submatrices $W_i \eqdef (\Wc{x}{y})_{\substack{x \in \Xc \\ y \in Y_i}}$ 
are symmetric.
A matrix is said to be symmetric if all if its rows are permutations of each other, and all its columns are permutations of each other.
\end{definition}

\noindent
{\em Remarks.}\\
\begin{itemize}
\item Such a channel is called {\em symmetric} in \cite[p.94]{G68}. We avoid using the same terminology as Gallager since ``symmetric channel'' is 
generally used now to denote a channel for which any row is a permutation of each other row and the same property also holds for the columns.
\item This notion is a generalization (when the output alphabet is discrete) of what is called a binary input symmetric channel in \cite{RU08}.
It also generalizes the notion of a cyclic symmetric channel in \cite{BB06}.
\item It is shown that for such channels \cite[Th. 4.5.2]{G68} a uniform distribution on the inputs maximizes the mutual information between the output and the input of the channel and gives therefore its capacity. In such a case, linear codes attain the capacity of such a channel.
\item This notion captures the notion of symmetry of a channel in a very broad sense. In particular the erasure channel is weakly symmetric 
(for many definitions of ``symmetric channels'' an erasure channel is not symmetric).
\end{itemize}

\begin{notation}
We denote for such a channel and for a given output $y$ by $\pi_y=(\pi(\alpha))_{\alpha \in \fq}$  the associated APP vector, that is $\pi(\alpha) = \prob(x=\alpha|y)$
where we denote by 
$x$ the input symbol to the channel. 
\end{notation}
To compute this APP vector we will make throughout the paper the following assumption
\begin{assumption}
The input of the communication channel is assumed to be uniformly distributed
over $\fq$.
\end{assumption}

We give now the asymptotic behavior of
the Koetter-Vardy decoder for a weakly symmetric channel, but before doing this we will need a few lemmas.
\begin{lemma}\label{lem:simple_formula}
Assume that $x$ is the input symbol that was sent and that the communication is weakly symmetric, then by viewing $\pi$ as a function of the random variable $y$ we have for any 
$x \in \Fq$:
$$
\esp_y(\pi(x)) = \esp_y\left(\norm{\pi}^2\right),
~~~~\hbox{ with } \norm{\pi}^2 \eqdef \sum_{\alpha \in \fq} \pi(\alpha)^2.$$
\end{lemma}

\begin{proof}
To prove this result, let us introduce some notation. Let us denote by
\begin{itemize}
\item
$\Yc$ the output alphabet and $ Y_1 \cup \dots \cup Y_n = \Yc$  is a partition of $\Yc$  such that all the submatrices $W_i \eqdef (\Wc{x}{y})_{\substack{x \in \Xc \\ y \in Y_i}}$ are symmetric for $i=1, \ldots, n$.
 \item $C_i \eqdef \sum_{x \in \Fq} \Wc{x}{y}$ and $C_i^{(2)} \eqdef \sum_{x \in \Fq} \Wc{x}{y}^2$ where $y$ is arbitrary in  $Y_i$ (these quantities do not depend
 on the element  $y$ chosen in $Y_i$);
 \item $R_i \eqdef \sum_{y \in Y_i} \Wc{x}{y}$  and $R_i^{(2)} \eqdef \sum_{y \in Y_i} \Wc{x}{y}^2$ where $x$ is arbitrary in $\Fq$.
\end{itemize}

We observe now that from the assumption that $x$ was uniformly distributed
\begin{equation}
\label{eq:begin}
\pi_y(\alpha) = \frac{\frac{1}{q}  \Wc{\alpha}{y}}{\prob(\text{receiving y})} = \frac{\frac{1}{q} \Wc{\alpha}{y} }{ \frac{1}{q}\sum_{\beta \in \Fq} \Wc{\beta}{y} } =
\frac{\Wc{\alpha}{y} }{\sum_{\beta \in \Fq} \Wc{\beta}{y} }.
\end{equation}
We observe now that 
\begin{equation*}
\esp_y(\pi(x)) =  \sum_{y \in \Yc} \pi_y(x) \Wc{x}{y} = 
\sum_{y \in \Yc} \frac{\Wc{x}{y} }{\sum_{\beta \in \Fq} \Wc{\beta}{y} } \Wc{x}{y}  
= \sum_{i=1}^n \sum_{y \in Y_i} \frac{\Wc{x}{y}^2 }{\sum_{\beta \in \Fq} \Wc{\beta}{y} } 
= \sum_{i=1}^n \frac{R_i^{(2)} }{C_i } \label{eq:pi_x}.
\end{equation*}
where the second equality is due to \eqref{eq:begin}.

On the other hand
\begin{eqnarray}
\esp_y\left(\norm{\pi}^2\right) &= & \sum_{y \in \Yc} \sum_{\alpha \in \Fq} \pi_y(\alpha)^2  \Wc{x}{y}  
 =  \sum_{y \in \Yc} \sum_{\alpha \in \Fq}  \frac{\Wc{\alpha}{y}^2 }{\left(\sum_{\beta \in \Fq} \Wc{\beta}{y}\right)^2 } \Wc{x}{y}  \nonumber\\
& = &  \sum_{i=1}^n \sum_{y \in Y_i} \sum_{\alpha \in \Fq} \frac{\Wc{\alpha}{y}^2 }{\left(\sum_{\beta \in \Fq} \Wc{\beta}{y}\right)^2 } \Wc{x}{y}  
 =   \sum_{i=1}^n \sum_{y \in Y_i} \sum_{\alpha \in \Fq} \frac{\Wc{\alpha}{y}^2 }{(C_i)^{2} } \Wc{x}{y}  \nonumber \\
& = &  \sum_{i=1}^n \sum_{y \in Y_i} \frac{C_i^{(2)}}{(C_i)^2 } \Wc{x}{y}  
 =   \sum_{i=1}^n \frac{R_i C_i^{(2)}}{(C_i)^2 } \label{eq:pi2}
\end{eqnarray}
where the second equality is due to \eqref{eq:begin}.

By summing all the elements (or the square of the elements) of the symmetric matrix $W_i$ either by columns or by rows and 
since all these row sums or all these column sums are equal, we obtain that
\begin{equation*}
\sum_{\alpha \in \Fq, y \in Y_i} \Wc{\alpha}{y} = |Y_i| C_i  =  q R_i
\end{equation*}
and 
\begin{equation*}
\sum_{\alpha \in \Fq, y \in Y_i} \Wc{\alpha}{y}^2 = |Y_i| C_i^{(2)}  
=  q R^{(2)}_i
\end{equation*}

By using these two equalities in \eqref{eq:pi2} we obtain
\begin{equation*}
\esp_y\left(\norm{\pi}^2\right) = \sum_{i=1}^n \frac{\frac{C_i |Y_i|}{q}  \frac{R_i^{(2)} q}{|Y_i|}}{(C_i)^2 }
=  \sum_{i=1}^n \frac{R_i^{(2)} }{C_i }
\end{equation*}
This yields the same expression as the one for $\esp_y(\pi(x))$ given in \eqref{eq:pi_x}. 
\end{proof}

As we will now show, this quantity $\esp\left(\norm{\pi}^2\right)$ turns out to be the limit of the rate for which the Koetter-Vardy 
decoder succeeds in decoding when the alphabet gets large. 
For this reason, we will denote this quantity by the {\em Koetter-Vardy capacity} of the channel.
\begin{definition}[Koetter-Vardy capacity]
Consider a weakly symmetric channel and denote by $\pi$ the associated probability vector.
The Koetter-Vardy capacity of this channel, which we denote by $\CKV$, is defined by
$$
\CKV \eqdef \esp(\norm{\pi}^2).
$$
\end{definition}

To prove that this quantity captures the rate at which the Koetter-Vardy is successful (at least for large lengths and therefore large field size) let us first prove concentration results around the expectation for the numerator and denominator appearing in the left-hand term of \eqref{eq:success}.

\begin{lemma}\label{lem:concentration}
Let $\epsilon > 0$ and $\mu \eqdef \esp( \norm{\pi}^2)$. We have
\begin{eqnarray}
\prob\left( \left\langle \Pim, \lfloor \mathbf 0 \rfloor \right\rangle \leq (1-\epsilon)  \mu n \right) 
 &\leq & e^{-2n \mu^2\epsilon^2}\label{eq:concentration_num}\\
\prob\left( \left\langle \Pim, \Pim \right\rangle \geq (1+\epsilon)  n \mu \right) 
&\leq&  e^{-2n \mu^2\epsilon^2}\label{eq:concentration_den}
\end{eqnarray}
\end{lemma}
\begin{proof}
Let us first prove \eqref{eq:concentration_den}. 
We can write the left-hand  term as a
sum of $n$ i.i.d. random variables
$$
\left\langle \Pim, \Pim\right\rangle = \sum_{j=1}^n X_j,
$$
where $X_j \eqdef \norm{\Pim(j)}^2$.
Note that (i) $\esp \left(X_j\right) = \esp(\norm{\pi}^2)$, (ii) $0 \leq X_j \leq 1$.
By using Hoeffding's inequality 
we obtain that for any $\epsilon>0$ we have
\begin{equation}
\prob\left( \sum_{j=1}^n X_j  \geq n  \mu (1+\epsilon) \right) \leq e^{-2n \mu^2\epsilon^2}.
\end{equation}

Now \eqref{eq:concentration_num}  can be dealt with in a similar way by writing
$$
\left\langle \Pim, \lfloor \zero \rfloor \right\rangle = \sum_{j=1}^n Y_j
$$
where $Y_j \eqdef \Pi^j(0)$. The channel is assumed to be symmetric and we can therefore use Lemma \ref{lem:simple_formula}
from which we deduce that $\esp (Y_j ) = \esp( \norm{\pi}^2) = \mu$.
We also have $0 \leq Y_j \leq 1$ and by applying Hoeffding's inequality we obtain that for any $\epsilon>0$ we have
\begin{equation}
\prob\left( \sum_{j=1}^n Y_j \leq n  \mu (1-\epsilon) \right) \leq e^{-2n \mu^2\epsilon^2}.
\end{equation}
\end{proof}

This result can be used to derive a rather tight upper-bound on the probability of error of the Koetter-Vardy decoder.
\begin{theorem}\label{th:exponential}
Consider a weakly symmetric $q$-ary input channel of Koetter-Vardy capacity $\CKV$.
Consider a Reed-Solomon code over $\Fq$ of length $n$, dimension $k$ such that its rate $R= \frac{k}{n}$ satisfies $R < \CKV$.
Let 
\begin{eqnarray*}
\delta & \eqdef &  \frac{\CKV - R}{R} \\
R^* & \eqdef & \frac{k-1}{n} \\
L & \eqdef & \left\lceil \frac{3 \left( \frac{1}{R^*}+\frac{\sqrt{q}}{2\sqrt{R^*}} \right)(1+\tfrac{\delta}{3})} {\delta} \right\rceil
\end{eqnarray*}
The probability that the Koetter-Vardy decoder with list size bounded by $L$ does not output in its list the right codeword is upper-bounded
by $\OO{e^{-K \delta^2 n}}$
for some constant $K$.
\end{theorem}
\begin{proof}
Without loss of generality we can assume that the all-zero codeword $\zero$ was sent.
From Theorem \ref{th:loose},  we know that the Koetter-Vardy decoder succeeds if and only if the following condition is met
$$
\frac{\left\langle \Pim, \lfloor \mathbf 0 \rfloor \right\rangle}{\sqrt{\left\langle \Pim, \Pim\right\rangle}} \geq 
\frac{\sqrt{k-1}}{1 - \frac{1}{L}\left( \frac{1}{R^*}+\frac{\sqrt{q}}{2\sqrt{R^*}} \right)}.
$$
Notice that the right-hand side satisfies 
\begin{equation}
\label{eq:lower-bound-rhs}
\frac{\sqrt{k-1}}{1 - \frac{1}{L}\left( \frac{1}{R^*}+\frac{\sqrt{q}}{2\sqrt{R^*}} \right) } \leq  
\frac{\sqrt{k-1}}{ 1 - \frac{\delta \left( \frac{1}{R^*}+\frac{\sqrt{q}}{2\sqrt{R^*}}\right)}{3 \left( \frac{1}{R^*}+\frac{\sqrt{q}}{2\sqrt{R^*}} \right)\left(1+\tfrac{\delta}{3}\right)}} =  
\frac{\sqrt{k-1}}{1- \frac{\delta}{3+\delta}} 
 =  \sqrt{k-1}\left(1+\tfrac{\delta}{3}\right) 
\end{equation}

Let $\epsilon$ be a positive constant that we are going to choose afterward.
Define the events $\Ec_1$ and $\Ec_2$ by 
\begin{eqnarray*}
\Ec_1 & \eqdef & \{\Pim: \left\langle \Pim, \Pim\right\rangle \leq n \CKV (1+\epsilon)\}\\
\Ec_2 & \eqdef & \{\Pim: \left\langle \Pim, \lfloor \zero \rfloor \right\rangle \geq n \CKV (1-\epsilon)\}.
\end{eqnarray*}
Note that
by Lemma \ref{lem:concentration} the events $\Ec_1$ and $\Ec_2$ have both probability $\geq 1 - \epsilon'$ 
where $\epsilon' \eqdef e^{-2n \CKV^2\epsilon^2}$.

Thus, the probability that event $\Ec_1$ and event $\Ec_2$ both occur is
$$\prob(\mathcal E_1 \cap \mathcal E_2) = \prob(\mathcal E_1) + 
\prob(\mathcal E_2) - \prob(\mathcal E_1 \cup \mathcal E_2) \geq  1 - \epsilon' + 1 - \epsilon' -1
= 1-2\epsilon'.$$
In the case $\Ec_1$ and $\Ec_2$ both hold, we have
\begin{equation}\label{eq:final2}
\frac{\left\langle \Pim, \lfloor \mathbf 0 \rfloor \right\rangle}{\sqrt{ \left\langle \Pim, \Pim \right\rangle}}
\geq \frac{1-\epsilon}{\sqrt{1+\epsilon}} \sqrt{ \CKV n} 
\end{equation}
A straightforward computation shows that for any $x>0$ we have 
$$
\frac{1-x}{\sqrt{1+x}} \geq 1-\frac{3}{2}x.
$$
Therefore for $\epsilon  >0$ we have in the aforementioned case
\begin{equation*}
\frac{\left\langle \Pim, \lfloor \mathbf 0 \rfloor \right\rangle}
{\sqrt{\left\langle \Pim, \Pim \right\rangle}}  \geq \left(1- \tfrac{3}{2} \epsilon\right) \sqrt{\CKV n} 
= \left(1- \tfrac{3}{2} \epsilon\right) \sqrt{(1+\delta)R n}
= \left(1- \tfrac{3}{2} \epsilon\right) \sqrt{k (1+ \delta) }
\end{equation*}

Let us choose now $\epsilon$ such that 
\begin{equation}
\label{eq:choice-epsilon}
\left(1-\tfrac{3}{2}\epsilon \right)\sqrt{1+\delta} = 1 + \tfrac{\delta}{3}.
\end{equation}
Note that $\epsilon = \Theta(\delta)$.
This choice implies that
\begin{equation*}
\frac{\left\langle \Pim, \lfloor \mathbf 0 \rfloor \right\rangle}
{\sqrt{\left\langle \Pim, \Pim \right\rangle}}  \geq  \sqrt{k} \left(1+\tfrac{\delta}{3}\right)  
 \geq  \sqrt{k-1} \left(1+\tfrac{\delta}{3}\right)
 \geq   \frac{\sqrt{k-1}}{1 - \frac{1}{L}\left( \frac{1}{R^*}+\frac{\sqrt{q}}{2\sqrt{R^*}} \right) }
\end{equation*}
where we
used in the last inequality the bound given in \eqref{eq:lower-bound-rhs}.

In other words, the Koetter Vardy decoder outputs the codeword $\zero$ in its list.
The probability that this does not happen is at most 
$2e^{-2n \CKV^2\epsilon^2}=e^{-n \Theta(\delta^2)}$.
\end{proof}

An immediate corollary of this theorem is the following result that gives a (tight) lower bound on the 
error-correction capacity of the Koetter-Vardy decoding algorithm over a discrete memoryless channel.
\begin{corollary}\label{cor:KVP}
Let $(\code{C}_n)_{n \geq 1}$ be an infinite family of Reed-Solomon codes of rate $\leq R$. Denote by $q_n$ the alphabet size
of $\code{C}_n$ that is assumed to be a non decreasing sequence that goes to infinity with $n$.
Consider an infinite family of $q_n$-ary weakly symmetric channels  with associated probability error vectors $\pi_n$ such 
that $\esp\left( \norm{\pi_n}^2 \right)$ has a limit as $n$ tends to infinity.
Denote by $\CKV^\infty$ the asymptotic Koetter-Vardy capacity of these channels, i.e.
$$
\CKV^\infty \eqdef \lim_{n \rightarrow \infty}  \esp\left( \norm{\pi_n}^2 \right).
$$
This infinite family of codes  can be decoded correctly 
by the Koetter-Vardy decoding algorithm with probability $1-o(1)$ 
as $n$ tends to infinity as soon as there exists 
$\epsilon >0$ such that 
$$
R \leq  \CKV^\infty -\epsilon.
$$
\end{corollary}

\begin{remark}
\label{Rem-UV-1}
Let us observe that for the $q\hbox{-SC}_{p}$ we have
$$
\esp\left( \norm{\pi}^2 \right) = (1-p)^2 +(q-1)\frac{p^2}{(q-1)^2} = (1-p)^2 + \Oq.
$$ 
By letting $q$ going to infinity, we recover in this way the performance of the Guruswami-Sudan algorithm which works as soon
as 
$R < (1-p)^2$.
\end{remark}

\noindent
\par{\bf Link with the results presented in \cite{KV03} and \cite{KV03a}.}
In \cite[Sec. V.B eq. (32)]{KV03} an arbitrarily small upper bound on the error probability $P_e$ is given, it is namely explained that 
$P_e \leq \epsilon$ as soon as the rate $R$ and the length $n$ of the Reed-Solomon code satisfy
$\sqrt{R} \leq \esp(\Zc^*) - \frac{1}{\sqrt{\epsilon n}}$ (where the expectation is taken with respect to the 
{\em a posteriori} probability distribution of the codeword). Here $\Zc^*$ is some function of the
multiplicity matrix which itself depends on the received word. This is not a bound of the same form as the one given in Theorem 
\ref{th:exponential} whose upper-bound on the error probability only depends on some well defined quantities which govern the complexity 
of the algorithm (such as the size $q$ of the field over which the Reed-Solomon code is defined and a bound on the list-size) and the Koetter-Vardy
capacity of the channel. 

However, many more details are given in the preprint version \cite{KV03a} of \cite{KV03} in Section 9.
There is for instance implicitly in the proof of Theorem 27 in \cite[Sec. 9]{KV03a}  an upper-bound on the error probability
of decoding a Reed-Solomon code with the Koetter-Vardy decoder which goes to zero polynomially fast in the length 
as long as the rate is less than $C \eqdef \frac{\trace \left(W (\diag P_{Y})^{-1} W^T \right)}{q^2}$ where $W$ is the transition probability matrix of the channel
and $\diag P_Y$ is the $|\Yc| \times |\Yc|$ matrix which is zero except on the diagonal where the diagonal elements give the probability distribution of the
output of the channel when the input is uniformly distributed. It is readily verified that in the case of a weakly symmetric channel $C$ is nothing but 
the Koetter-Vardy capacity of the channel defined here. $C$ can be viewed as a more general definition of the ``capacity'' of a channel adapted to the Koetter-Vardy decoding algorithm.
 However it should be said that ``error-probability'' in \cite{KV03,KV03a} should be understood here as ``average error probability of error''
where the average is taken over the set of codewords of the code. It should be said that this average may vary wildly among the codewords in the case of a non-symmetric channel.
In order to avoid this, we have chosen a different route here and have assumed some weak form of symmetry for the channel which ensures that the probability of error
does not depend on the codeword which is sent. The authors of \cite{KV03a} use a second moment method to bound the error probability, this can only give polynomial
upper-bounds on the error probability. This is why we have also used a slightly different route in Theorem \ref{th:exponential}  to obtain stronger (i.e. exponentially small) upper-bounds on the error probability.

\section{Algebraic-soft decision decoding of AG codes.}
\label{SectionAG}

The problem with Reed-Solomon codes is that their length is limited by the alphabet size.
To overcome this limitation it is possible to proceed as in \cite{KV03a} and use instead Algebraic-Geometric codes (AG codes in short)
which can also be decoded by an extension of the Koetter-Vardy algorithm and which have more or less a similar error correction capacity 
as Reed-Solomon codes under this decoding strategy.
The extension of this decoding algorithm to AG codes is sketched in Section \ref{sec:KV_AG}.
Let us first recall how these codes are defined.

An AG code is constructed  from a  triple $(\X, \P, mQ)$ where:
\begin{itemize}
\item $\X$ denotes an algebraic curve over a finite field $\Fq$
(we refer to \cite{S93a} for more information about algebraic geometry codes);
\item $\mathcal P=\left\{ P_1, \ldots, P_n\right\}$ denotes a set of $n$ distinct points of $\X$ with coordinates in $\Fq$; 
\item $mQ$ is a divisor of the curve, here $Q$ denotes another point in $\X$ with coordinates in $\Fq$ which is not in $\mathcal P$
and $m$ is a nonnegative integer.
\end{itemize}

We define $\mathcal L(mQ)$ as the vector space of rational functions on $\X$ that  
may contain only a pole at $Q$ and the multiplicity of this pole is at most $m$.
Then, the \emph{algebraic geometry}  code associated to the above triple denoted by $\C{X}{P}{mQ}$ is the image of $\mathcal L(mQ)$ under the evaluation map $\begin{array}{cccc}\mathrm{ev}_{\mathcal P}: & \mathcal L(mQ) & \longrightarrow & \mathbb F_q^n\end{array}$ defined by $\mathrm{ev}_{\mathcal P}(f) = \left(f(P_1), \ldots , f(P_n) \right)$, i.e.
$$\C{X}{P}{mQ} \eqdef \left\{ \mathrm{ev}_{\mathcal P}(f)=\left(f(P_1), \ldots, f(P_n) \right) \mid f\in \mathcal L(mQ)\right\}$$
Since the evaluation map is linear,  the code $\C{X}{P}{mQ}$ is a linear code of length $n$ over $\Fq$ and dimension $k=\dim(\mathcal L(mQ))$.
This dimension can be lower bounded by $k \geq m-g+1$ where 
$g$ is the genus of the curve. Recall that this quantity is defined by 
$$
g \eqdef \max_{m \geq 0} \{ m - \dim \Lc (mQ) \} +1
$$
Moreover the minimum distance $d$ of this code satisfies $d \geq n - m$.

Reed-Solomon codes are a particular case of the family of AG codes and correspond to the case where
$\X$ is the affine line over $\Fq$, $\P$ are $n$ distinct elements of $\Fq$ and $\Lc$ is the vector space of polynomials of degree at most $k-1$ and with coefficients in $\Fq$.

Recall that it is possible to obtain for any designed rate $R=\tfrac{k}{n}$ and any square prime power $q$ an infinite family of AG codes over $\Fq$ of rate  $ \geq R$ 
of increasing length $n$ and minimum distance $d$ meeting ``asymptotically'' the MDS bound as $q$ goes to infinity
$$ \frac{d}{n}  \geq  (1-R) - O\left(\frac{1}{\sqrt{q}}\right)$$ 
This follows directly from the two aforementioned lower bounds $k \geq m-g+1$ and $d \geq n-m$ and the well known result of Tsfasman, Vl{\u{a}}duts and Zink \cite{TVZ82}
\begin{theorem}[\cite{TVZ82}]
\label{th:TVZ}
For any number $R \in [0,1]$ and any square prime power $q$ there exists an infinite family of AG codes over $\Fq$ of rate  $ \geq R$ 
of increasing length $n$ such that the normalized genus $\gamma \eqdef \frac{g}{n}$ of the underlying curve  satisfies
$$
\gamma \leq \frac{1}{\sqrt{q}-1}
$$
\end{theorem}
We will call such codes {\em Tsfasman-Vl{\u{a}}duts-Zink} AG codes in what follows.

As is done in \cite{KV03a}, it will be helpful to assume that $2g-1 \leq m < n$. This implies among other things that the 
dimension of the code is given my $k=m-g+1$. 
\ire{$k=m-g+1$}
We will make this assumption from now on. 
As in \cite{KV03} it is possible to obtain a soft-decision list decoder with a list which does not exceed some prescribed quantity $L$. Similar to the Reed-
Solomon case considered in \cite{KV03}, it suffices to increase the value of $s$ in \cite{KV03}[Algorithm A] until we get a matrix $\Mm$ such that
$L< L_m(\Mm) < L+1$, where $L(\Mm)$ is a bound on the list of the codewords output by the algorithm
which is given in Lemma \ref{AG-1}, and then to use this matrix $\Mm$ in the Koetter Vardy decoding algorithm.

The following result is similar to \cite[Th. 17]{KV03} 
\begin{theorem}\label{th:listsize_AG}
Algebraic soft-decoding for AG codes with list-size limited to $L$ produces a list that contains a codeword $\mathbf c\in \C{X}{P}{mQ}$ if 
\begin{equation}
\label{eq:listsize_AG}
\frac{\left\langle \Pim, \lfloor \mathbf c\rfloor\right\rangle}{\sqrt{\left\langle \Pim, \Pim\right\rangle}}
\geq \sqrt{m} 
\frac{1+ \frac{\tilde{\gamma}+\sqrt{2 \tilde{\gamma}}}{L \sqrt{1-\tfrac{2\tilde{\gamma}}{L}\left( 1+\tfrac{2}{L}\right)}}}{1 - \frac{1}{L \sqrt{1-\tfrac{2 \tilde{\gamma}}{L}\left( 1+\tfrac{2}{L}\right)}}
\left( \frac{\sqrt{q }}{2 \sqrt{\tilde{R}}} +   \frac{1}{\tilde{R}} \right)} = \sqrt{m}\left(1 + \OO{\frac{1}{L}} \right)
\end{equation}
where $\tilde{R}=\frac{m}{n}$, $\tilde{\gamma} = \frac{g}{m}$ and $\mathcal O(\cdot)$ depends only on $\tilde{R}$, $\tilde{g}$ and $q$.
\end{theorem}
The proof of this theorem can be found in Section \ref{sec:KV_AG} of the appendix. It heavily relies on results proved in the preprint version 
\cite{KV03a} of \cite{KV03}.

\begin{theorem}\label{th:exponentialAG}
Consider a weakly symmetric $q$-ary input channel of Koetter-Vardy capacity $\CKV$ where $q$ is a square prime power.
Consider a Tsfasman-Vl{\u{a}}duts-Zink AG code over $\Fq$ of length $n$, dimension $k$ such that its rate $R= \frac{k}{n}$ satisfies $R < \CKV - \gamma$ where $\gamma \eqdef \frac{1}{\sqrt{q}-1}$.
Let 
\begin{eqnarray*}
\delta & \eqdef &  \frac{\CKV - R - \gamma}{R} \\
\tilde{R} & \eqdef & \frac{m}{n} \\
\tilde{\gamma} & \eqdef & \frac{g}{m}\\
f(\ell) & \eqdef & \frac{1+ \frac{\tilde{\gamma}+\sqrt{2 \tilde{\gamma}}}{\ell \sqrt{1-\tfrac{2\tilde{\gamma}}{\ell}\left( 1+\tfrac{2}{\ell}\right)}}}{1 - \frac{1}{\ell \sqrt{1-\tfrac{2 \tilde{\gamma}}{L}\left( 1+\tfrac{2}{\ell}\right)}}
\left( \frac{\sqrt{q }}{2 \sqrt{\tilde{R}}} +   \frac{1}{\tilde{R}} \right)} \\
L & \eqdef & f^{-1}\left(1+\frac{\delta}{3} \right) \end{eqnarray*}
The probability that the Koetter-Vardy decoder with list size bounded by $L$ does not output in its list the right codeword is upper-bounded
by $$\OO{e^{-K \delta^2 n}}
$$
for some constant $K$. Moreover $L = \Th{1/\delta}$ as $\delta$ tends to zero.
\end{theorem}

\begin{proof}
The proof follows word by word the proof of Theorem \ref{th:exponential} with the only difference that 
$k-1$ is replaced by $m=k+g$. The only new ingredient is that 
we use Theorem \ref{th:listsize_AG} instead of \eqref{eq:loose} which explains the 
new form chosen for the list-size $L$. 
The last part, namely that $L = \Th{1/\delta}$ is a simple consequence of the fact that 
$f(L) = 1 + \Th{\frac{1}{L} }$ as $L$ tends to infinity.
\end{proof}

\section{Correcting errors beyond the Guruswami-Sudan bound}
\label{sec:iteratedUV}
The purpose of this section is to show that the $\UV$ construction improves significantly the noise level that the
Koetter-Vardy decoder is able to correct. 
To be more specific, consider the $q$-ary symmetric channel. 
The asymptotic Koetter-Vardy capacity of a family of $q$-ary symmetric channels of crossover probability $p$ is equal to $(1-p)^2$.  
It turns out that this 
is also the maximum
crossover probability that the Guruswami-Sudan decoder is able to sustain when the alphabet and the length go to infinity.
We will prove here 
that the $\UV$ construction with 
Reed-Solomon components already performs a 
bit better than 
$(1-p)^2$
when the rate is small enough.
By using iterated $\UV$ constructions we will be able to improve rather significantly the performances and this even for a moderate number of levels.

Our analysis of the Koetter-Vardy decoding is done for weakly symmetric channels. When we want to analyze a $\UV$ code based on Reed-Solomon codes
used over a channel $W$ it will be helpful that the channels $W^0$ and $W^1$ viewed by the decoder of $U$ and $V$ respectively are also weakly symmetric.
Simple examples show that this is not necessarily the case. However a slight restriction of 
the notion of weakly symmetric channel considered in \cite{BB06} does the job. It consists in the notion of a cyclic-symmetric channel whose definition is given below.

\begin{definition}[cyclic-symmetric channel]
\label{def:strictly_cyclic_symmetric}
We denote for a vector $\yv=(y_i)_{i \in \Fq}$ with coordinates indexed by a finite field $\Fq$ by 
$\yv^{+g}$ the vector $\yv^{+g}=(y_{i+g})_{i \in \Fq}$, by $n(\yv)$ the number of $g$'s in $\Fq$ such that 
$\yv^{+g} = \yv$ and by $\yv^*$ the set $\{\yv^{+g}, g \in \Fq\}$. A $q$-ary input channel 
is cyclic-symmetric channel 
if and only
there exists a probability function $Q$ defined over the sets of possible $\pi^*$ such that for 
any $i \in \Fq$ we have
$$
\prob(\pi = \yv|x=i) = y_i n(\yv) Q(\yv^*).
$$
\end{definition}

The point about this notion is that $W^0$ and $W^1$ stay cyclic-symmetric when $W$ is cyclic-symmetric and
that a cyclic-symmetric channel is also weakly symmetric.
This will allow to analyze the asymptotic error correction capacity of 
iterated $\UV$ constructions.
\begin{proposition}[\cite{BB06}]
\label{pr:csc}
Let $W$ be a cylic-symmetric channel. Then $W$ is weakly symmetric and $W^0$ and $W^1$ are also cyclic-symmetric.
\end{proposition}

\subsection{The $\UV$-construction}
We study here how a $\UV$ code performs when $U$ and $V$ are both 
Reed-Solomon
codes decoded with the Koetter-Vardy decoding algorithm 
when the communication channel 
is a $q$-ary symmetric channel of error probability $p$.

\begin{proposition}
\label{UV-1}
For any real $p$ in $[0,1]$ and real $R$ such that 
$$R<C^\infty_{\UV}(p) \eqdef \frac{(p^3-4p^2+4p-4)(1-p)^2}{2(p-2)},$$
there exists an infinite family of $\UV$-codes of rate $ \geq R$ based on Reed-Solomon codes whose alphabet size $q$ increases with the length 
and whose probability of error on the $\qSCp$ when decoded by the iterated $\UV$-decoder based on the Koetter-Vardy decoding algorithm
goes to $0$ with the alphabet size.
\end{proposition}

\begin{proof}
The $(U_0|U_0+U_1)$-construction can be decoded correctly by the Koetter-Vardy decoding algorithm if
it decodes correctly $U_0$ and $U_1$. 
Let $\pi_i$ be the APP probability vector seen by the decoder for $U_i$ for $i \in \{0,1\}$.
A $\qSCp$ is clearly a cyclic-symmetric channel and therefore the channel viewed by the $U_0$ decoder and the
$U_1$ decoder are also cyclic-symmetric by Proposition \ref{pr:csc}. A cyclic-symmetric channel is weakly symmetric and therefore
by Corollary \ref{cor:KVP}, decoding succeeds with probability $1-o(1)$ when we choose the rate $R_i$ of $U_i$ to be any 
positive number below $\lim_{q\rightarrow \infty}\esp\left( \norm{\pi_{i}}^2 \right)$ for $i \in \{0,1\}$.

In Section \ref{Appendix-UV} of the appendix it is proved in Lemmas \ref{L12} and \ref{L11} that
\begin{eqnarray*}
\esp\left( \norm{\pi_{0}}^2 \right) &=&\frac{(p+2)(p-1)^2}{2-p} +  \Oq \\
\esp\left( \norm{\pi_{1}}^2 \right) &=& (1-p)^4 + \Oq
\end{eqnarray*} 

 Since the rate $R$ of the $\UV$ construction is equal to $\frac{R_0+R_1}{2}$ 
decoding succeeds with probabilty $1 - o(1)$ if
$$R< \lim_{q \rightarrow \infty} \frac{
\esp \left(\norm{\pi_{0}}^2 \right)  
+ 
\esp \left(\norm{\pi_{1}}^2 \right)
}{2} 
= \frac{(p^3-4p^2+4p-4)(1-p)^2}{2(p-2)}. $$
\end{proof}

From Figure \ref{TwiceUV} we deduce that the $\UV$ decoder outperforms the RS decoder with
Guruswami-Sudan or Koetter-Vardy decoders
as soon as $R<0.17$.

\subsection{Iterated $\UV$-construction}
Now we will study what happens over a $q$-ary symmetric channel with error probability $p$ if we apply 
the iterated $\UV$-construction with Reed-Solomon codes as constituent codes. 
In particular, the following result handles the cases of the iterated $\UV$-construction of depth $2$ and $3$.

\begin{proposition}
\label{UV-2-3}
For any real $p$ in $[0,1]$ we define
\begin{equation}
\label{Depth2}
C^{\infty,(2)}_{\UV}(p) \eqdef \frac{Q(p)(1-p)^2}{4(p^2- 2p + 2)(3p - 4)(2-p)^2}
\end{equation}
and 
\begin{equation}
\label{Depth3}
C^{\infty,(3)}_{\UV}(p) \eqdef \frac{S(p)(1-p)^2}{T_1(p)
T_2(p)T_3(p)
\left(3p^2 - 6p + 4\right)\left(p^2 - 2p + 2\right)^2\left(7p - 8\right)\left(3p - 4\right)^2\left(2-p\right)^4}
\end{equation}

Then, for any real $R$ such that $R<C^{\infty,(2)}_{\UV}(p)$ (resp. $R<C^{\infty,(3)}_{\UV}(p)$) there exists an infinite family of iterated $\UV$-codes of depth $2$ (resp. of depth $3$) and rate $\geq R$ based on Reed-Solomon codes whose alphabet size $q$ increases with the length and whose probability of error 
with 
the Koetter-Vardy decoding algorithm goes to $0$ with the alphabet size.

Where

{\tiny
\begin{eqnarray*}
Q(p) &\eqdef & 3p^{11} - 40p^{10} + 243p^9 - 890p^8 + 2192p^7 - 3800p^6 +
4702p^5 - 4148p^4 + 2624p^3 - 1248p^2 + 480p - 128,\\
T_1(p) & \eqdef & p^4 - 4p^3 + 6p^2 - 4p + 2,\\
T_2(p)& \eqdef & 7p^2 - 18p + 12, \\
T_3(p)& \eqdef &  5p^2 - 12p+ 8 \hbox{ and }\\
S(p) &\eqdef & 
6615p^{35} - 269766p^{34} + 5348715p^{33} - 68697432p^{32} +
642499307p^{31} - 4663447618p^{30} + 27338551153p^{29} \\
& - & 
133009675740p^{28}+ 547673160274p^{27} - 1936548054764p^{26} + 
5946432348816p^{25} -15994984917120p^{24} \\
& + & 
37947048851166p^{23} - 79831430926900p^{22}  + 149553041935846p^{21} 
- 250287141028584p^{20} + 375085789739404p^{19} \\
& - & 
504157479736392p^{18} + 608316727420536p^{17}  -  659027903954592p^{16} 
+ 640716590979968p^{15} - 558310438932224p^{14} \\
& + & 435164216863552p^{13} - 302519286136704p^{12} + 186871196449024p^{11} 
- 102093104278528p^{10} + 49062052366336p^9 \\
& - & 
20617356455936p^8 + 7534906109952p^7 - 2386429566976p^6 
+ 655237726208p^5 - 156569829376p^4 \\
&+&
32471121920p^3 - 5628755968p^2 + 723517440p - 50331648
\end{eqnarray*}
}
\end{proposition}

\begin{proof}
The proof is given in Appendix \ref{Appendix-UV2} and \ref{Appendix-UV3} for iterated $\UV$ construction of depth $2$ and $3$, respectively.
\end{proof}

Figure \ref{TwiceUV} summarizes the performances of these iterated $\UV$-constructions
From this figure we see that if we apply the iterated $\UV$-construction of depth $2$ we get better performance than decoding a classical Reed-Solomon code with the Guruswami-Sudan decoder for low rate codes, specifically for $R<0.325$. Moreover, if we apply the iterated $\UV$-construction of depth $3$ we get even better results, we beat the Guruswami-Sudan for codes of rate $R< 0.475 $.

\begin{figure}[h!]
\centering
\includegraphics[scale=0.7]{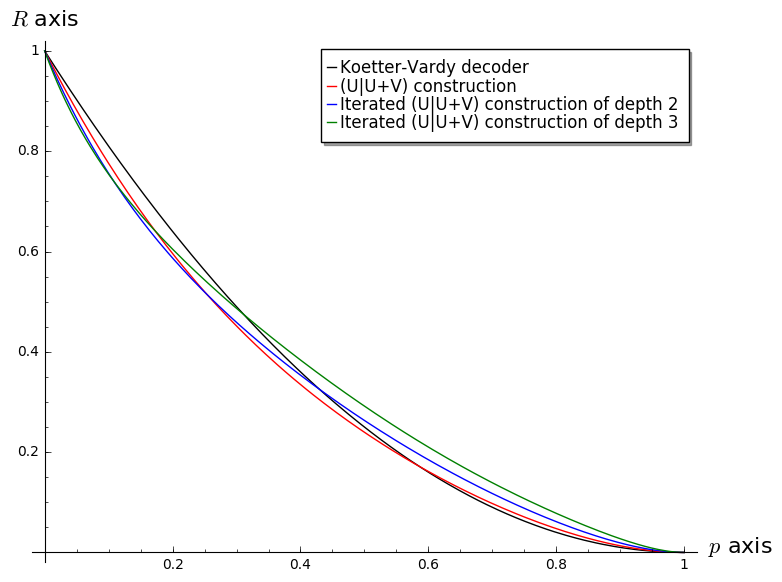}
\caption{Rate plotted against the crossover error probability $p$ for four code-constructions. The black line refers to standard Reed-Solomon codes decoded by the Guruswami-Sudan algorithm, the red line to the $\UV$-construction, the blue line to the iterated $\UV$-construction of depth $2$ and the green line to the iterated $\UV$-construction of depth $3$. \label{fig:infinite}}
\label{TwiceUV}
\end{figure}

\subsection{Finite length capacity}
\label{ss:finite_capacity}

Even if for finite alphabet size $q$ the Koetter-Vardy capacity cannot be understood as a capacity in the usual sense: no family of codes is known
which could be decoded with the Koetter-Vardy decoding algorithm and whose probability of error would go to zero as the codelength goes to infinity 
at any rate below the Koetter-Vardy capacity. Something like that is only true  approximately for AG codes when the size of the alphabet is a
square prime power and if we are willing to pay an additional term of $\frac{1}{\sqrt{q}-1}$ in  the gap between the Koetter-Vardy 
capacity and the actual  code rate. Actually, we can even be sure that for certain rates this result can not hold, since
the Koetter-Vardy capacity can be above the Shannon capacity for very noisy channels. Consider for instance 
the ``completely-noisy''  $q$-ary symmetric channel of crossover probability $\frac{q-1}{q}$.
Its capacity is $0$ whereas its Koetter-Vardy capacity is equal to $\frac{1}{q}$. Nevertheless it is still insightful to consider  
$f(W,\ell) \eqdef \frac{1}{2^\ell} \sum_{i=0}^{2^\ell -1} \CKV(W^i)$ where $W^i$ is the channel viewed by the constituent $U_i$ code for an 
iterated-$UV$ construction of depth $\ell$ for a given noisy channel. This could be considered as the limit for which we can not 
hope to have small probabilities of error after decoding when using Reed-Solomon codes constituent  codes and 
the Koetter-Vardy decoding algorithm. We have plotted these functions in Figure \ref{fig:finite_capacity} for $q=256$ and $\ell =0$ up to $\ell=6$ and 
a $\qSCp$. It can be seen
that for $\ell=5,6$ we get rather close to the actual capacity of the channel in this way.

\begin{figure}[h!]
\centering
\includegraphics[scale=0.7]{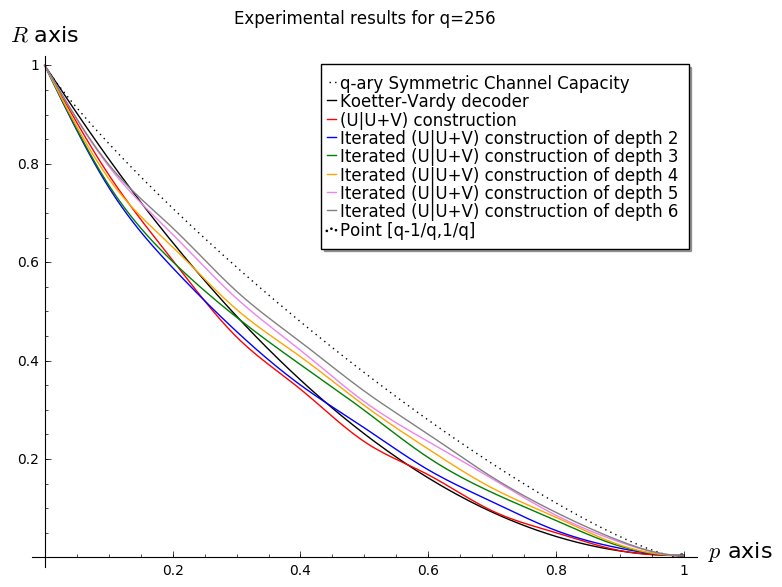}
\caption{average Koetter-Vardy capacity plotted against the crossover error probability $p$ for seven code constructions. 
The noise model is a 
a $\qSCp$.  \label{fig:finite_capacity}.}
\end{figure}

\section{Attaining the capacity with an iterated $\UV$ construction}
\label{sec:capacity}
When the number of levels for which we iterate this construction tends to infinity, we attain the capacity of any 
$q$-ary symmetric channel at least when the cardinality $q$ is prime. This is a straighforward consequence of the fact that polar codes
attain the capacity of any $q$-ary symmetric channel. Moreover the probability of error after decoding can be made to be 
almost exponentially small with respect to the overall codelength.
More precisely the aim of this section is to prove the following results about the probability of error.
 
 \begin{theorem}\label{th:errprob1}
Let $W$ be a cyclic-symmetric $q$-ary channel where $q$ is prime. Let $C$ be the capacity of this channel.
There exists $\epsilon_0 >0$ such that for any $\epsilon$ in the range $(0,\epsilon_0)$ and any $\beta$ in the range $(0,1/2)$ there exists a sequence of iterated $\UV$ codes with Reed-Solomon constituent codes of arbitrarily large length  which have  rate $\geq C - \epsilon$ when the codelength is sufficiently large and whose probability of error  $P_e$ is upper bounded by
$$
P_e \leq n e^{-\frac{K(\epsilon,\beta) N}{n \log N}}
$$
when decoded with the iterated $\UV$ decoder based on decoding the constituent codes with the 
Koetter-Vardy decoder with listsize bounded by $\OO{\frac{1}{\epsilon}}$ and 
where $N$ is the codelength, $n = O(\log \log N) ^{1/\beta}$, 
and $K(\epsilon, \beta)$ is some positive function of $\epsilon$ and $\beta$.
\end{theorem}

For the iterated $\UV$-construction with algebraic geometry codes as constituent codes  we obtain an even stronger result which is
\begin{theorem}\label{th:errprob2}
Let $W$ be a cyclic-symmetric $q$-ary channel where $q$ is prime. Let $C$ be the capacity of this channel.
There exists $\epsilon_0 >0$ such that for any $\epsilon$ in the range $(0,\epsilon_0)$ there exists a sequence of iterated $\UV$ codes of arbitrarily large length with AG defining codes of rate $\geq C - \epsilon$  when the codelength is sufficiently large and  whose probability of error  $P_e$ is upper bounded by
$$
P_e \leq e^{-K(\epsilon) N}
$$
when decoded with the iterated $\UV$ decoder based on the Koetter-Vardy algorithm
 with listsize bounded by $\OO{\frac{1}{\epsilon}}$ 
and
where $K$ is some positive function of $\epsilon$.
\end{theorem}

\noindent
{\em Remarks:}
\begin{itemize}
\item
In other words the exponent of the error probability is in the first case (that is with Reed-Solomon codes) almost of the form 
$-\frac{K(\epsilon) N}{\log  N (\log \log N)^{2+\epsilon}}$ where $\epsilon$ is an arbitrary positive constant. This is significantly better than the concatenation of polar codes with Reed-Solomon codes (see 
\cite[Th. 1]{BJE10}
and also \cite{MELK14} for some more practical variation of this construction)
which leads to an exponent of the form
$- K(\epsilon) \frac{N}{\log^{27/8} N}$.
\item The second case leads to a linear exponent and is therefore optimal up to the dependency in $\epsilon$.
\item Both results are heavily based on the fact that when the depth of the construction tends to infinity the channels
viewed by the decoders at the leaves of the iterated construction polarize: they have either capacity close to $1$ or close to $0$. 
This follows from a generalization of Ar{\i}kan's polarization result on binary input channels.
This requires $q$ to be prime. However it is possible to change slightly the $\UV$ structure in order to have polarization for all alphabet sizes.
Taking for instance in the case where $q$ is a prime power at each node instead of the $\UV$ construction a random $(U|U+\alpha V) = \{(\uv|\uv+\alpha \vv):\uv \in U,\;\vv \in V\}$ where
$\alpha$ is chosen randomly in $\Fq^\times$ would be sufficient here for ensuring polarization of the corresponding channels and would ensure that our results on the probability or error of the 
iterated construction would also work in this case.
\item The reason why these results do not capture the dependency in $\epsilon$ of the exponent comes from 
the fact that only rather rough results on polarization are used  (we rely namely on 
Theorem \ref{th:polarization}). Capturing the dependency on $\epsilon$ really needs much more precise results
on polarization, such as for instance finite length scaling of polar codes. This will be discussed in the next section.
\end{itemize}

\par{\bf Overview of the proof of these theorems.}
The proof of these theorems uses  four ingredients.
\begin{enumerate}
\item The first ingredient is the polarization theorem \ref{th:polarization}. It shows that when the number of levels of the recursive $\UV$ construction tends to infinity, the fraction of the decoders of the constituent codes who  face an almost noiseless channel tends to the capacity of the original channel.
Here the measure for being noisy is the Bhattacharyya parameter of the channel.
\item We then show   that when the Bhattacharyya parameter is close to $0$ the Koetter Vardy capacity of the channel is close to $1$ meaning that 
we can use Reed-Solomon codes or AG codes of rate close to $1$ for those almost noiseless constituent codes (see Proposition \ref{pr:Bha_KV}).
\item When we use Tsfasman-Vl{\u{a}}duts-Zink AG codes and if $q$ were allowed to be a square prime power, the situation would be really clear. 
For the codes in our  construction that face an almost noiseless channel, we use as constituent AG codes Tsfasman-Vl{\u{a}}duts-Zink AG codes of rate 
of the form $1 - \epsilon - \frac{1}{\sqrt{q}-1}$.
This gives an exponentially small (in the length of the constituent code) probability of error for each of those constituent codes by using 
Theorem  \ref{th:exponentialAG}.
For the other codes, we just use the zero code (i.e. the code with only the all-zero codeword). Now in order to get an exponentially small 
probability of error, it suffices to take the number of levels to be large (but fixed !) so that the fraction of almost noiseless channels is close enough to  capacity and to let the 
length of the constituent codes go to infinity. This gives an exponentially small probability of error  when the rate is bounded away from 
capacity by a term of order $\frac{1}{\sqrt{q}-1}$. 
\item In order to get rid of this term, and also in order to be able to use Tsfasman-Vl{\u{a}}duts-Zink AG codes for the case we are interested in, namely an alphabet which is prime, we use another argument.
Instead of using a $q$-ary code over a $q$-ary input channel we will use a $q^m$-ary code over this $q$-ary input channel.
In other words, we are going to group the received symbols by packets of size $m$ and view this as a channel
with $q^m$-ary input symbols. This changes the Koetter Vardy capacity of the channel. It turns out that the Koetter-Vardy 
capacity of this new channel is the Koetter-Vardy capacity of the original channel raised to the power $m$
(see Proposition \ref{pr:KVm}). This implies that when the Koetter-Vardy capacity was close to $1 - \epsilon$, the new Koetter-Vardy 
capacity is close to $1 - m \epsilon$ and we do not lose much in terms of capacity when moving to a higher alphabet.
This allows to use AG codes over a higher alphabet in order to get arbitrarily close to capacity by still keeping
an exponentially small probability of error (we can indeed take $m$ fixed but sufficiently large here). For Reed-Solomon 
codes, the same trick works and allows to use constituent codes of arbitrarily large length by making the alphabet grow with 
the length of the code. However in this case, we can not take $m$ fixed anymore and this is the reason why we lose a little
bit in the behavior of the error exponent. Moreover the number of levels is also increasing in the last case in order to make
the Bhattacharyya parameter sufficiently small at the almost noiseless constituent codes so that the Koetter-Vardy 
stays sufficiently small after grouping symbols together.
\end{enumerate}

\par{\bf Link between the Bhattacharyya parameter and the Koetter-Vardy capacity.}
We will provide here a proposition showing that for a fixed alphabet size the Koetter Vardy capacity of a channel $W$ is greater than
$1 - (q-1)\Bha{W}$. For this purpose, it will be helpful to use  
 an alternate form of the Bhattacharyya parameter
\begin{equation}
\label{eq:Bha}
 \Bha{W} =  \sum_{y \in \Yc} \prob(Y=y) Z(X|Y=y)
\end{equation}
where $X$ is here a uniformly distributed random variable, $Y$ is the output corresponding to sending $X$ over 
the channel $W$ and
\begin{equation}
\label{eq:Z}
Z(X|Y=y) \eqdef  \frac{1}{q-1} \sum_{x,x' \in \F_q\\ x' \neq x} \sqrt{\prob(X=x|Y=y)} \sqrt{\prob(X=x'|Y=y)}
\end{equation}

\begin{proposition}\label{pr:Bha_KV}
For a symmetric channel
$$C_{\text{KV}}(W) \geq 1 - (q-1)\Bha{W}$$
\end{proposition}

\begin{proof}
To simplify formula here we will write $p(x|y)$ for $\prob(X=x|Y=y)$, $p(x)$ for $\prob(X=x)$ and 
$p(y)$ for $\prob(Y=y)$.
The proposition is essentially a consequence of the well known fact that the R\'enyi entropy which is defined for all $\alpha >0$, $\alpha \neq 1$ by
$$
H_\alpha(X) \eqdef \frac{1}{1-\alpha} \log_q \sum_x p(x)^\alpha
$$
and 
$$H_1(X) = \lim_{\alpha \rightarrow 1} H_\alpha(X)$$
(which turns out to be equal to the usual Shannon entropy taken to the base $q$)
is decreasing in $\alpha$. 
This also holds of course for the ``conditional'' R\'enyi entropy which is defined by 
$$
H_\alpha(X|Y=y) \eqdef \frac{1}{1-\alpha} \log_q \sum_x p(x|y)^\alpha
$$
Consider now a random variable $X$ which is uniformly distributed over $\Fq$ and let $Y$ 
be the corresponding output of the memoryless channel $W$.
By using the definition of the Bhattacharyya parameter given by \eqref{eq:Bha} 
we can write 
$$
 \Bha{W} = \sum_{y \in \Yc} p(y) Z(X|Y=y)
$$
where
$$
Z(X|Y=y) =  \frac{1}{q-1} \sum_{x,x' \in \F_q, x' \neq x} \sqrt{p(x|y)} \sqrt{p(x'|y)}
$$
We observe that we can relate this quantity to the R\'enyi entropy of order $\frac{1}{2}$ through
\begin{eqnarray}
H_{1/2}(X|Y=y) & = & 2 \log_q \sum_x \sqrt{p(x|y)} \nonumber \\
& = & \log_q \left( \sum_x \sqrt{p(x|y)} \sum_{x'} \sqrt{p(x'|y)} \right) \nonumber\\
& = & \log_q \left( \sum_x p(x|y) + \sum_{x,x' \in \F_q, x' \neq x} \sqrt{p(x|y)} \sqrt{p(x'|y)} \right) \nonumber\\
& = & \log_q \left( 1 +  (q-1) Z(X|Y=y) \right) \label{eq:one}
\end{eqnarray}

On the other hand we know that $H_2(X|Y=y) \leq H_{1/2}(X|Y=y)$. 
Recall that 
$$
H_2(X|Y=y) = -\log_q \sum_x p(x|y)^2
$$
Using this together with \eqref{eq:one} we obtain that 
\begin{equation}
\label{eq:two}
- \log_q \sum_x p(x|y)^2 \leq \log_q \left( 1 +  (q-1) Z(X|Y=y) \right)
\end{equation}
Let 
$$
S \eqdef 1- \sum_x p(x|y)^2
$$
Observe that 
\begin{equation}
- \log_q \sum_x p(x|y)^2  =  \log_q \left( \frac{1}{ \sum_x p(x|y)^2} \right) 
 =  \log_q \frac{1}{1-S} 
 \geq   \log_q (1+S) \label{eq:three}
\end{equation}

Finally by using \eqref{eq:two} together with \eqref{eq:three} we deduce that 
$\log_q (1+S) \leq   \log_q \left( 1 + (q-1) Z(X|Y=y) \right)$
which implies that
\begin{equation}
1- \sum_x p(x|y)^2 = S  \leq  (q-1) Z(X|Y=y)
\end{equation}
By averaging over all $y$'s we get
\begin{equation}
\sum_y p(y) \left( 1- \sum_x p(x|y)^2 \right)  \leq (q-1) \sum_y p(y) Z(X|Y=y) 
\end{equation}
This implies the proposition by noticing that
\begin{eqnarray*}
\sum_y p(y) \left( 1- \sum_x p(x|y)^2 \right) & = & 1 - C_{\text{KV}}\\
(q-1) \sum_y p(y) Z(X|Y=y) & = & (q-1) \Bha{W}.
\end{eqnarray*}
\end{proof}

\par{\bf Changing the alphabet.}
The problem with Reed-Solomon codes is that their length is bounded by their alphabet size. It would be desirable
to have more freedom in choosing their length. There is a way to overcome this difficulty by grouping together transmitted symbols into packets 
 and to view each packet as a symbol over a larger alphabet.
In other words, assume that we have a memoryless communication channel $W$ with input alphabet $\Fq$. Instead of looking for codes defined
over $\Fq$ we will group input symbols in packets of size $m$ and view them as  symbols in the extension field 
$\Fqm$. This will allow us to consider codes defined over $\Fqm$ and allows much more freedom in choosing the length of
the Reed-Solomon codes components (or more generally the AG components). There is one caveat to this approach, it is that we change the channel model. In such a 
case the channel is $W^{\otimes m} \eqdef \underbrace{W \otimes W \otimes \dots W}_{m \text{ times}}$ where
we define the tensor of two channels by 
\begin{definition}[Tensor product of two channels] Let 
$W$ and $W'$ be two memoryless channels with respective input alphabets $\Xc$ and $\Xc'$
and respective output alphabets $\Yc$ and $\Yc'$. Their tensor product $W \otimes W'$ is a 
memoryless channel with input alphabet $\Xc \times \Xc'$ and output alphabet 
$\Yc \times \Yc'$ where the transitions probabilities are given by
$$
W \otimes W'(y,y'|x,x') = W(y|x) W(y'|x)
$$
for all $(x,x',y,') \in \Xc \times \Xc' \times \Yc \times \Yc'$.
\end{definition}

The Koetter-Vardy capacity of this tensor product is easily related to the Koetter-Vardy capacity of the initial channel through
\begin{proposition}\label{pr:KVm}
$\CKV(W^{\otimes m}) = \CKV(W)^m$.
If $\CKV(W) = 1 - \epsilon$ then $\CKV(W^{\otimes m}) \geq 1 - m \epsilon.$
\end{proposition}
\begin{proof}
Let $\xv = (x_1 \dots x_m \in \Fq^m)$ 
be the sent symbol for channel $W^{\otimes m}$ and
$\yv = (y_1,\dots,y_m)$ be the received vector.
Let $\pi^m_{\yv}$ be the APP probability vector after receiving $\yv$, that is
$\pi^m_{\yv} = (\prob(\xv=(\alpha_1, \dots, \alpha_m)|\yv)_{(\alpha_1, \dots, \alpha_m) \in \Fqm}$.
We denote the $(\alpha_1, \dots, \alpha_m)$ component of this vector by $\pi^m_{\yv}(\alpha_1, \dots, \alpha_m)$.
Let $\pi_{y_i}=(\prob(x_i = \alpha|y_i))_{\alpha \in \Fq}$ be the APP vector for the $i$-th use of the channel.
We denote by $\pi_{y_i}(\alpha)$ the $\alpha$ component of this vector.
Observe that 
\begin{equation}
\pi^m_{\yv}(\alpha_1 \dots \alpha_m) = \pi_{y_1}(\alpha_1) \dots \pi_{y_m}(\alpha_m).
\end{equation}
This implies that 
$$
\norm{\pi^m_{\yv}}^2 = \norm{\pi_{y_1}}^2 \dots \norm{\pi_{y_m}}^2
$$
This together with the fact that the channel is memoryless implies that
\begin{eqnarray*}
\CKV(W^{\otimes m}) & = & \esp\left( \norm{\pi^m_{\yv}}^2\right)) \\
& = & \esp\left(\norm{\pi_{y_1}}^2\right) \dots \esp\left(\norm{\pi_{y_m}}^2 \right) \\
& = & \CKV(W)^m 
\end{eqnarray*}
The last statement follows easily from this identity and the convexity inequality $(1-x)^m \geq 1 - mx$ 
which holds for
$x$ in $[0,1]$ and $m \geq 1$.
\end{proof}

\par{\bf Proof of Theorem \ref{th:errprob1}.} We have now all ingredients at hand for proving Theorem \ref{th:errprob1}.
We use Theorem \ref{th:polarization} to claim that there exists a lower bound $\ell_0$ on the number of levels $\ell$ in a recursive 
$\UV$ construction such that
\begin{equation}\label{eq:weak_polarization}
 \frac{1}{n} \left| i \in \{0,1\}^\ell : \Bha{W^i} \leq 2^{-n^\beta} \right|  \geq C - \epsilon/2
\end{equation}
for all $\ell \geq \ell_0$ where $n \eqdef 2^\ell$. We call the channels that satisfy this condition the {\em good channels}.
We choose our code to be a recursive $\UV$-code of depth $\ell$ with Reed-Solomon  constituent codes that are of length $q^m$ and defined over $\Fqm$. The overall length (over $\Fq$) of the 
recursive $\UV$ code is then 
\begin{equation}
N \eqdef 2^\ell m q^m.
\end{equation}
The constituent codes $U^i$ that face a good channel $W^i$ are chosen as Reed-Solomon codes of  
dimension $k$ given by
$$
k = \left\lfloor q^m \left(1 - m (q-1)2^{-n ^\beta} - \frac{\epsilon}{4} \right) \right\rfloor 
$$
whereas all the other codes are chosen to be zero codes. By using Proposition \ref{pr:Bha_KV} we know that 
$$
\CKV\left(W^i\right) \geq 1- (q-1)2^{-n^\beta}.
$$
From this we deduce that the Koetter-Vardy of the channel corresponding to grouping   together $m$ symbols in $\Fq$ has 
a Koetter-Vardy capacity that satisfies 
 $$\CKV\left[\left(W^i\right)^{\otimes m} \right] \geq 1 -  m(q-1)2^{-n^\beta}.$$
  Now we can invoke Theorem \ref{th:exponential} 
and deduce that the probability of error of the Reed-Solomon codes that face these good channels 
when decoding them with the Koetter-Vardy decoding algorithm with list size bounded
by $\OO{\frac{1}{\epsilon}}$ is upper-bounded by a quantity of the
form $e^{-K q^m \epsilon^2}$. The overall probability of error is bounded 
by $n e^{-K q^m \epsilon^2}$.
 
 We choose now $n$ such that it is the smallest power of two for which the inequality
$$
m(q-1)2^{-n^\beta} \leq \epsilon/4
$$
holds. This implies $n =0(\log ^{1 / \beta} m)$ as $m$ tends to infinity.
Since $\frac{k}{q^m} = 1 - \epsilon/2 -o(1)$ as $m$ tends to infinity, the rate of the iterated $\UV$ code is of order 
$(C - \epsilon/2)(1- \epsilon/2-o(1)) = C  - \frac{1+C}{2} \epsilon + \epsilon^2/4 + o(1) \geq C - \epsilon$
 for $\epsilon$ sufficiently small and $n$ sufficiently large when $C<1$. When $C=1$ the theorem is obviously true.
 This together with the previous upper-bound on the probability of a decoding error imply directly our theorem since $N  = n m q^m$ and $n= 0(\log ^{1 / \beta} m)$ imply that 
$m = \frac{\log N(1+o(1))}{\log q}$ as $m$ tends to infinity.
\par{\bf Proof of Theorem \ref{th:errprob2}.}
Theorem \ref{th:errprob2} uses similar arguments, the only difference is that now the number of levels $\ell$ in the construction and the parameter 
$m$ only depend 
on the gap to capacity we are looking for. We fix $\beta$ to be an arbitrary constant in $(0,1/2)$ and  choose $m$ to be the smallest {even integer} $m$ for which $\frac{1}{\sqrt{q^m}-1}$ is smaller than $\epsilon/4$ and the number of levels $\ell$ to be the smallest integer such that we have at the same time
\begin{eqnarray}
 \frac{1}{n} \left| i \in \{0,1\}^\ell : \Bha{W^i} \leq 2^{-n^\beta} \right| & \geq &C - \epsilon/4  \label{eq:polar1}\\
 & \text{and} & \nonumber \\
 m(q-1)2^{-n^\beta}& \leq &\epsilon/4 \label{eq:mne}
\end{eqnarray}
where $n \eqdef 2^\ell$. Such an $\ell$ necessarily exists by Theorem \ref{th:polarization}.

We choose our code to be a recursive $\UV$-code of depth $\ell$ with 
Tsfasman-Vl{\u{a}}duts-Zink AG constituent codes that are of length $N_0$ and defined over $\Fqm$. 
Such codes exist by the Tsfasman-Vl{\u{a}}duts-Zink construction for arbitrarily large lengths because $m$ is even.
The overall length (over $\Fq$) of the 
recursive $\UV$ code is then 
\begin{equation}
N \eqdef 2^\ell m N_0.
\end{equation}
For the constituent codes $U^i$ that face a good channel $W^i$, we choose the rate of the 
AG code to be  $\frac{k}{N_0}$ where 
$$
k = \left\lfloor N_0 (1 - m (q-1)2^{-n ^\beta} - \frac{1}{\sqrt{q^m}-1} - \frac{\epsilon}{4} ) \right\rfloor 
$$
whereas all the other codes are chosen to be zero codes.
The rate of the codes that face a good channel is clearly greater than or equal a quantity 
of the form $1-3\epsilon/4 +o(1)$ as $N_0$ goes to infinity by using \eqref{eq:mne} and $\frac{1}{\sqrt{q^m}-1} \leq \epsilon/4$.
The overall rate $R$ of the iterated $\UV$ code satisfies therefore $R \geq (C- \epsilon/4) (1-3\epsilon/4 +o(1)) \geq C - \frac{3C+1}{4}\epsilon + 3\epsilon^2/4 + o(1) \geq C - \epsilon $ for $\epsilon$ sufficiently small and $N_0$ sufficiently large when $C<1$. We can make the assumption $C<1$ from now on, since when $C=1$ the theorem is trivially true.

On the other hand, the error probability of decoding a code $U^i$  facing a good channel $W^i$ 
with the Koetter-Vardy decoding algorithm with list size bounded
by $\OO{\frac{1}{\epsilon}}$ is upperbounded by a quantity of the form
$e^{-K N_0 \epsilon^2}$ by using Theorem \ref{th:exponentialAG} since the rate $R_0$ of such a code satisfies
\begin{eqnarray*}
R _0 & \leq & 1- m (q-1)2^{-n ^\beta} - \frac{1}{\sqrt{q^m}-1} - \frac{\epsilon}{4} +o(1) \\
& \leq & \CKV({(W^i)}^{\otimes m}) -   \frac{1}{\sqrt{q^m}-1} - \frac{\epsilon}{4} +o(1)
\end{eqnarray*}
by using the lower bound on the Koetter-Vardy capacity of a good channel that follows from Propositions \ref{pr:Bha_KV} and
\ref{pr:KVm}. 
 The overall probability of error is therefore bounded 
by $n e^{-K N_0 \epsilon^2}$. This probability is of the form announced in Theorem \ref{th:errprob2} since $m$ and $n$ are 
quantities that only depend on $q$ and $\epsilon$.

\section{Conclusion}
\par{\bf A variation on polar codes that is much more flexible.} We have given here a variation of polar 
codes that allows to attain capacity with a polynomial-time decoding complexity in a more flexible way than
standard polar codes. It consists in taking an iterated-$\UV$ construction based on Reed-Solomon codes 
or more generally AG codes. Decoding consists in computing the APP of each position in the same way 
as polar codes and then to decode the constituent codes with a soft information decoder, the Koetter-Vardy 
list decoder in our case. Polar codes are indeed a special case of this construction by taking constituent codes
that consist of a single symbol. However when we take constituent codes which are longer we benefit from the fact 
that we do not face a binary alternative as for polar codes, i.e. putting information or not in the symbol depending on the
noise model for this symbol, but can choose freely the length (at least in the AG case) and the rate of the constituent code
that face this noise model.

{\bf An exponentially small probability of error.} This allows to control  the rate and error probability in a much finer way as for standard polar codes.
Indeed  the failure probability of polar codes is essentially governed by the error probability of an information  symbol of the
polar code facing the noisiest channel (among all information symbols).
In our case, this error probability can be decreased significantly by choosing a long enough code and a rate below
the noise value that our decoder is able to sustain (which is more or less the Koetter-Vardy capacity of the
noisy channel in our case). Furthermore, now we can also put information in channels that were not used for sending  
information in the polar code case. 
When using Reed-Solomon codes with this approach we obtain a quasi-exponential decay 
of the error probability which is 
significantly smaller than for the standard concatenation of an inner polar code with 
an outer Reed-Solomon code. When we use AG codes we even obtain an exponentially fast decay of the probability of error 
after decoding. 

The whole work raises a certain
number of intriguing questions.

{\bf Dependency of the error probability with respect to the gap to capacity.}
Even if the exponential decay with respect to the codelength of the iterated $\UV$-construction is optimal,
the result says nothing about the behavior of the exponent in terms of the gap to capacity. The best we can hope
for is a probability of error which behaves as $e^{-K \epsilon^2 n}$ where $\epsilon$ is the gap to capacity, that is
$\epsilon = C - R$, $C$ being the capacity and $R$ the code rate. We may observe that Theorem \ref{th:exponential}
gives a behavior of this kind with the caveat that $\epsilon$ is not the gap to capacity there but the gap to the Koetter-Vardy
capacity. To obtain a better understanding of the behavior of this exponent, we need to have a much finer understanding of  the speed 
of polarization than the one given in Theorem \ref{th:polarization}. What we really need is indeed a result of the following form
\begin{equation}\label{eq:strong_polarization}
 \frac{1}{n} \left| i \in \{0,1\}^\ell : \Bha{W^i} \leq \epsilon \right|  \geq C - f(\epsilon,\ell)
\end{equation}
which expresses the fraction of ``$\epsilon$-good'' channels in terms of the gap to capacity with sharp estimates 
for the ``gap'' function $f(\epsilon,\ell)$. 
The problem in our case is that our understanding of the speed of polarization is far from being complete.
Even for  binary input channels, the information  we have on the function $f(\epsilon,\ell)$ is only partial as shown by \cite{HAU14,GB14,GX15,MHU16}.
A better understanding of the speed of polarization could then be used in order to get a better understanding of the 
decay of the error probability in terms of the gap to capacity. A tantalizing issue is whether or not we get a better scaling than for polar codes.

{\bf Choosing other kernels.} The iterated $\UV$-construction can be viewed  as choosing
the original polar codes from Arikan associated to the kernel
 $\Gm = \begin{pmatrix} 1 & 0 \\1 & 1 \end{pmatrix}$. Taking larger kernels  does not improve the error probability after decoding 
 in the binary case for polar codes, unless taking very large kernels as shown in \cite{K09b}, however this is not the case for non binary kernels.
 Even ternary kernels, such as for instance  the ternary ``Reed-Solomon'' kernel $\Gm_{\text{RS}} = \begin{pmatrix} 1 &  1 &0 \\1 & -1 & 1 \\ 1 & 1 & 1 \end{pmatrix}$ results
 in a better behavior of the probability of error after decoding (see \cite{MT14}). 
 This raises the issue whether other kernels would allow to obtain better results in our case. In other words, would other
 generalized concatenated code constructions do better in our case ?
Interestingly enough, it is not necessary a Reed-Solomon kernel which gives the best results in our case.
Preliminary results seem to show that it should be better to take
the kernel $\Gm = \begin{pmatrix} 1 & 0 & 0  \\1 & 1 & 0 \\ 1 & 1 & 1 \end{pmatrix}$
rather than the aforementioned Reed-Solomon kernel.
The last kernel would correspond to a $\UVW$ construction which is defined by
$$
\UVW = \left\{ (\mathbf u\mid\mathbf u + \mathbf v\mid \uv + \vv + \wv): \mathbf u \in U,\;\mathbf v \in V\text{ and } \wv \in W\right\}.
$$
One level of concatenation outperforms the $\UV$ construction on the $\qSCp$. In particular it allows to increase the slope at the origin of the
``infinite Koetter-Vardy capacity curve'' (when compared to the curve for one level in Figure \ref{fig:infinite}). This seems to be the key 
for choosing good kernels. This issue requires further studies.

{\bf Practical constructions.}
We have explored here the theoretical behavior of this coding/decoding strategy. 
What is suggested by the experimental evidence shown in Subsection \ref{ss:finite_capacity} is that these codes do not only have some theoretical
significance, but that they should also yield interesting codes for practical applications. Indeed 
Figure \ref{fig:finite_capacity} shows that it should be possible to get very close 
to the channel capacity by using only a construction with  a small depth, say $5-6$ together with constituent codes of moderate length
that can be chosen to be Reed-Solomon codes (say codes of length a hundred/a few hundred at most). This raises many issues that we 
did not cover here, such as for instance
\begin{itemize}
\item to choose appropriately the code rate of each constituent code in order to maximize the overall rate 
with respect to a certain target error probability;
\item choose the multiplicities for each constituent code in order to attain a good overall tradeoff complexity vs. performance;
\item choose other constituent codes such as AG codes  especially in cases where the channel is an $\Fq$-input channel for small values of
$q$. It might also be worthwhile to study the use of subfield subcodes of Reed-Solomon codes in this setting (for instance
BCH codes).
\end{itemize}
The whole strategy leads to use Koetter-Vardy decoding for Reed-Solomon/AG  codes in a regime where the rate gets either 
very close to $0$ or to $1$. This could be exploited to lower the complexity of generic Koetter-Vardy decoding.

\appendix

\section*{Appendix notation and assumption}
\label{sec:appendix_notation}
Throughout the appendix we will use the same notation as in Section \ref{sec:polar} and denote by $U_i$ a constituent
code of an iterated $\UV$ construction of some depth $\ell$. Here $i$ is an $\ell$-bit word. We assume in the appendix that all the constituent codes
are Reed-Solomon codes and that the model of error is the $q$-ary symmetric channel of crossover probability $p$ ($\qSCp$). 
We also denote by $\pi_i$ the APP probability vectors 
computed for decoding $U_i$. 
Without loss of generality we may assume that the codeword which is sent is the $0$ codeword. 
With this assumption, in order to reduce the number of cases to be considered it will be very convenient to give these APP vectors only up to a permutation
acting on all positions with the exception of the first one which will always be fixed. We will namely use the following notation.
\begin{notation}
For two probability vectors $\pv$ and $\pv'$ in $\R^q$ we will write
$\pv \eqs \pv'$ if and only if $p(0)=p'(0)$ and $p(1),\dots,p(q-1)$ is a permutation of
$p'(1),\dots,p'(q-1)$.
\end{notation}
Moreover also in order to simplify the expressions wich will appear in these APP vectors we will use the following notation
\begin{notation}
$\varepsilon^t$ denotes an arbitrary function of $q$ which satisfies $|\varepsilon^t| = \OO{\frac{1}{q^t}}$.
\end{notation}

\section{The $\UV$ construction}
\label{Appendix-UV}

With the zero codeword assumption, the distribution of the APP vector $\piv$ of a $\qSCp$-channel  is  as follows
$\piv = \left(1-p,\frac{p}{q-1},\frac{p}{q-1},\dots,\frac{p}{q-1}\right)$ with probability $1-p$ and 
$\piv = \left(\frac{p}{q-1},\dots,\frac{p}{q-1},1-p,\frac{p}{q-1},\dots,\frac{p}{q-1}\right)$ with probability $\frac{p}{q-1}$ where the term $1-p$ appears in an arbitrary position
with the exception of the first one. We summarize this in Table \ref{tab:level0}.
\begin{table}[h!]
\caption{\label{tab:level0} }
\begin{center}
\begin{tabular}{|c|c|c|}
\hline 
$\piv$ & probability & $\norm{\piv}^2$   \\
\hline $p_0^0  = \left(1-p,\frac{p}{q-1},\frac{p}{q-1},\dots,\frac{p}{q-1}\right)$ & $1-p$ & $(1-p)^2 + \varepsilon$ \\
\hline $ p_0^1 \eqs \left(\frac{p}{q-1},1-p,\frac{p}{q-1} \dots \frac{p}{q-1}\right)
$ & $p$ & $(1-p)^2 +\varepsilon$ \\
\hline 
\end{tabular}
\end{center}
\end{table}

The distribution of $\piv \times \piv'$ where $\piv$ and $\piv'$ are independent and distributed as in Table \ref{tab:level0}
is given in Table \ref{tab:level1_times}.
\begin{table}[h!]
\caption{\label{tab:level1_times}}
\begin{center}
\begin{tabular}{|c|c|c|}
\hline 
$\piv_0 $ & prob. & $\norm{\piv}^2$   \\
\hline $\pv_1^0 =\left(1-\epsilon,\varepsilon^2,\dots,\varepsilon^2 \right)$ & $(1-p)^2$ & $1- \varepsilon$ \\
\hline $\pv_1^1 \eqs \left(A,A,B \dots B\right)$ & $2p(1-p)$ & $\left(\frac{1-p}{2-p}\right)^2 + \varepsilon$ \\
where $A=\frac{1-p}{2-p} + \varepsilon$ and $B=\frac{p}{(q-1)(2-p)}+ \epsilon^2$ & & \\
\hline $\pv_1^2 \eqs \left(B,A,A,B\dots B\right)$ & $p^2 + \varepsilon$ & $\left(\frac{1-p}{2-p}\right)^2 + \varepsilon$ \\
with $A$ and $B$ as in the previous case & & \\
\hline 
\end{tabular}
\end{center}
\end{table}

We have used in Table \ref{tab:level1_times} the following notation.
\begin{lemma}
\label{L12}
Let $\piv_{0}$ be the APP probability vector viewed by the decoder $U_0$. For the channel error model of the code $U_0$ we have
\begin{eqnarray*}
\esp\left( \norm{\piv_{0}}^2 \right) & = & \frac{(p+2)(p-1)^2}{2-p} +  \Oq
\end{eqnarray*}
\end{lemma}

The distribution of $\piv \oplus \piv$ where $\piv$ and $\piv'$ are independent and distributed as in Table \ref{tab:level0}
is given in Table \ref{tab:level1_oplus}.
\begin{table}[h!]
\caption{\label{tab:level1_oplus}}
\begin{center}
\begin{tabular}{|c|c|c|}
\hline 
$\piv_1 $ & prob. & $\norm{\piv}^2$   \\
\hline 
$\pv_1^3 = \left(1-p_1+\varepsilon, \frac{p_1}{q-1}+\varepsilon^2, \dots, \frac{p_1}{q-1}+\varepsilon^2\right)$ & $(1-p)^2$ & $(1-p)^2 +\varepsilon$ \\
where $p_1 = 2p-p^2$ & & \\
\hline 
$\pv_1^4 \eqs \left( \frac{p_1}{q-1}+\varepsilon^2, 1-p_1+\varepsilon, \frac{p_1}{q-1}+\varepsilon^2 \dots \frac{p_1}{q-1}+\varepsilon^2\right)$ & $2p-p^2+\varepsilon$ & $(1-p)^2 + \varepsilon$ \\
\hline 
\end{tabular}
\end{center}
\end{table}

\begin{lemma}
\label{L11}
Let $\piv_{1}$ be the APP probability vector viewed by the decoder $U_1$. 
The channel error model for the code $U_1$ is a $q\hbox{-SC}_{p_1+\varepsilon}$ with $p_1=2p-p^2$
and 
we have
\begin{eqnarray*}
\esp\left( \norm{\piv_{1}}^2 \right) & = (1-p)^4 + \Oq.
\end{eqnarray*}
\end{lemma}

\noindent
{\bf Important remark:}
Observe that for the distribution of $\piv \times \piv'$ and the distribution of 
$\piv \oplus \piv'$ we have implicitly used the fact that with probability $1 - \Oq=1-\epsilon$ 
the two vectors $\piv$ and $\piv'$ have their $1-p$ entry in a different position when 
$\piv$ and $\piv'$ both correspond to the second case of Table \ref{tab:level0}. 
This accounts for the $\epsilon$ term in the probability for the third case of Tables
\ref{tab:level1_times} and \ref{tab:level1_oplus}. Since we are interested in obtaining the 
expected values of $\esp\left(\norm{\piv_i}^2 \right)$ only up to $\Oq$ we can readily ignore the cases when 
$\piv$ and $\piv'$ have ther $1-p$ value at the same position (assuming that this is not the first position).
This reasoning will be used repeatedly in the following sections.

\section{The iterated $\UV$ construction of depth $2$}
\label{Appendix-UV2}

\subsection{Computation of $\esp\left( \norm{\piv_{00}}^2\right)$}
Note that $\piv_{00}=\piv \times \piv'$ where  $\piv$ and $\piv'$ are independent and distributed as $\piv_0$ 
which is given in Table \ref{tab:level1_times}. Table \ref{tab:level2_timestimes} summarizes this distribution.
\begin{table}[h!]
\caption{\label{tab:level2_timestimes}}
\begin{center}
\begin{tabular}{|c|c|c|}
\hline 
$\piv_{00} $ & prob. & $\norm{\piv}^2$   \\
\hline 
$ \pv_2^0 \eqdef \left\{\begin{array}{c}
\pv_1^0  \times \pv_1^0,
\pv_1^0  \times \pv_1^1,\\
\pv_1^0  \times \pv_1^2,
\pv_1^1  \times \pv_1^1 \end{array} \right\}= \left(1-\epsilon,\varepsilon^2,\dots,\varepsilon^2 \right)$ & $(1-p)^2(1-2p+3p^3) + \varepsilon$ & $1- \varepsilon$ \\
\hline 
$\pv_2^1 \eqdef \pv_1^1 \times \pv_1^2 \eqs \left(C,C,C,C,D \dots D\right)$ & $4p^2(1-p)^2+\varepsilon$ & 
$4\left(\frac{1-p}{4-3p}\right)^2 + \varepsilon$ \\
where $C=\frac{1-p}{4-3p}+\varepsilon$ and $D=\frac{p}{(q-1)(4-3p)}+\varepsilon^2$ & & \\
\hline 
$\pv_2^2 \eqdef \pv_1^2 \times \pv_1^2 \eqs \left(D, C,C,C,C, D \dots D\right)$ & 
$p^2 + \varepsilon$ & $4\left(\frac{1-p}{4-3p}\right)^2 + \varepsilon$ \\
with $C$ and $D$ as in the previous case & & \\
\hline 
\end{tabular}
\end{center}
\end{table}

\begin{lemma}
\label{L21}
Let $\piv_{00}$ be the APP probability vector viewed by the decoder $U_{00}$. For the channel error model of the code $U_{00}$ we have 
\begin{eqnarray*}
\esp\left( \norm{\piv_{00}}^2 \right) & = & \frac{(5p^3-6p^2-5p-4)(1-p)^2}{3p-4}+\Oq
\end{eqnarray*}
\end{lemma}

\subsection{Computation of $\esp\left( \norm{\piv_{01}}^2\right)$}
It can be observed that $\piv_{01}=\piv \oplus \piv'$ where  $\piv$ and $\piv'$ are independent and distributed as $\piv_0$ 
which is given in Table \ref{tab:level1_times}. Table \ref{tab:level2_timesoplus} summarizes this distribution.

\begin{table}[h!]
\caption{\label{tab:level2_timesoplus}}
\begin{center}
\begin{tabular}{|c|c|c|}
\hline 
$\piv_{01} $ & prob. & $\norm{\piv}^2$   \\
\hline 
$\pv_2^3  \eqdef \pv_1^0 \oplus \pv_1^0 = 
\left(1-\epsilon,\varepsilon^2,\dots,\varepsilon^2 \right)$ & $(1-p)^4$ & $1- \varepsilon$ \\
\hline
$ \pv_2^4 \eqdef \pv_1^0 \oplus \pv_1^1 \eqs
\left(E, E, F  \dots F \right)$ & $4p(1-p)^3$ & $2\left(\frac{1-p}{2-p}\right)^2+ \varepsilon$ \\
with $E=\frac{1-p}{2-p}+\varepsilon$ and $F=\frac{p}{(2-p)(q-1)}+\varepsilon^2$ & & \\
\hline
$\pv_2^5 \eqdef \pv_1^0 \oplus \pv_1^2 \eqs
\left(F, E, E, F \dots F \right)$ & $2p^2(1-p)^2+\varepsilon$ & $2\left(\frac{1-p}{2-p}\right)^2+ \varepsilon$ \\
with $E$ and $F$ as in the previous case & & \\
\hline
$\pv_2^6 \eqdef \pv_1^1 \oplus \pv_1^1 \eqs \left(G, G, G, G,H\dots H\right)$ & $4p^2(1-p)^2+\varepsilon$ &$4\left( \tfrac{1-p}{2-p}\right)^4+ \varepsilon$  \\
where $G=\left(\frac{1-p}{2-p}\right)^2+\varepsilon$ and $H=\frac{p(4-3p)}{(q-1)(2-p)^2}+\varepsilon^2$ & & \\
\hline 
$\pv_2^7 \eqdef \pv_1^1 \oplus \pv_1^2 \eqs \left(H, G,G,G,G,H\dots H\right)$ & $4p^2(1-p)^2+p^2+\varepsilon$ & 
$4\left( \tfrac{1-p}{2-p}\right)^4+ \varepsilon$ \\
with $G$ and $H$ as in the previous case & & \\
\hline 
\end{tabular}
\end{center}
\end{table}

\begin{lemma}
\label{L22}
Let $\piv_{01}$ be the APP probability vector viewed by the decoder $U_{01}$. For the channel error model of the code $U_{01}$ we have
\begin{eqnarray*}
\esp\left( \norm{\piv_{01}}^2 \right) & = & \frac{(2+p)^2(1-p)^4}{(2-p)^2}+\Oq
\end{eqnarray*}
\end{lemma}

\subsection{Computation of $\esp\left( \norm{\piv_{10}}^2\right)$ and $\esp\left( \norm{\piv_{11}}^2\right)$}
Note that $\piv_{10}$ and $\piv_{11}$ are distributed like $\piv \times \piv'$ and $\piv \oplus \piv'$ where
$\piv$ and $\piv'$ are independent and distributed like $\piv_1$ which is itself the APP vector
obtained from transmitting $0$ over a $q\hbox{-SC}_{p_1+\varepsilon}$ with $p_1=2p-p^2$.
We can therefore use directly both lemmas of the previous section and obtain
\begin{lemma}
\label{L234}
Let $\piv_{10}$ and $\piv_{11}$ be the APP probability vectors viewed by the decoder $U_{10}$ and $U_{11}$, respectively. For the channel error model of the code $U_{10}$ we have
\begin{eqnarray*}
\esp\left( \norm{\piv_{10}}^2 \right) & = & \frac{(2+p_1)(1-p_1)^2}{(2-p_1)}  + \Oq = 
\frac{(1-p)^4(-p^2+2p+2)}{(p^2-2p+2)}+\Oq \\
\end{eqnarray*}
The channel error model for the code $U_{11}$ is a $q\hbox{-SC}_{p_2+\varepsilon}$ with $p_2=2p_1-p_1^2=p(2-p)(p^2-2p+2)$
and 
we have
\begin{eqnarray*}
\esp\left( \norm{\piv_{11}}^2 \right) &= & (1-p_2)^2 + \Oq = (1-p)^8 + \Oq 
\end{eqnarray*}
\end{lemma}

\section{The iterated $\UV$ construction of depth $3$}
\label{Appendix-UV3}
\subsection{Computation of $\esp\left( \norm{\piv_{000}}^2\right)$}

Note that $\piv_{000} = \piv \times \piv'$ where $\piv$ and $\piv'$ are independent and distributed as $\piv_{00}$ which is given in Table \ref{tab:level2_timestimes}. Table \ref{tab:level2_timestimestimes} summarizes this distribution.

\begin{table}[h!]
\caption{\label{tab:level2_timestimestimes}}
\begin{center}
\begin{tabular}{|c|c|c|}
\hline 
$\piv_{000} $ & prob. & $\norm{\piv}^2$   \\
\hline 
$\pv_3^0= \left\{\begin{array}{c}
\pv_2^0 \times \pv_2^0, \pv_2^0 \times \pv_2^1\\
\pv_2^0 \times \pv_2^2, \pv_2^1 \times \pv_2^1 
\end{array}
\right\}   = \left(1-\varepsilon,\varepsilon^2,\dots,\varepsilon^2 \right)$ &
$1-8p^7(1-p)-p^8 + \varepsilon$
& $1- \varepsilon$ \\
\hline 
$\pv_3^1 = \pv_2^1 \times \pv_2^2 
 \eqs (\underbrace{I, \ldots, I}_{8 \text{-times}}, J, \dots, J )$ & 
$8p^7(1-p)+ \varepsilon$ & $8\left( \frac{1-p}{8-7p}\right)^2 + \varepsilon$ \\
with $I=\frac{1-p}{8-7p} + \varepsilon$ and $J=\frac{p}{(8-7p)(q-1)} + \varepsilon$ & & \\
\hline 
$\pv_3^2 = \pv_2^2 \times \pv_2^2 
 \eqs (J, \underbrace{I, \ldots, I}_{8 \text{-times}}, J, \dots, J )$ & 
$p^8 + \varepsilon$ & $8\left( \frac{1-p}{8-7p}\right)^2 + \varepsilon$ \\
with $I$ and $J$ as in the previous case & & \\
\hline
\end{tabular}
\end{center}
\end{table}

\begin{lemma}
Let $\piv_{000}$ be the APP probability vector viewed by the decoder $U_{000}$ decoder.
For the channel error model of the code $U_{000}$ we have
\begin{eqnarray*}
\esp\left( \norm{\piv_{000}}^2 \right) & = & 
\frac{-(41p^7-14p^6-13p^5-12p^4-11p^3-10p^2-9p-8)(1-p)^2}{(8-7p)} + \Oq
\end{eqnarray*}
\end{lemma}

\subsection{Computation of $\esp\left( \norm{\piv_{001}}^2\right)$}
It can be observed that $\piv_{001} = \piv \oplus \piv'$ where $\piv$ and $\piv'$ are independent and distributed as $\piv_{00}$ which is given in Table \ref{tab:level2_timestimes}. Table \ref{tab:level2_timestimesoplus} summarizes this distribution.

\begin{table}[h!]
\caption{It is easy to check that $\norm{\pv_3^3}^2 = 1-\varepsilon$, 
$\norm{\pv_3^4}^2 = 4C^2 + \varepsilon$ and
$\norm{\pv_3^5}^2 = 16K^2 + \varepsilon$  with $C$ and $K$ as defined in this table. \label{tab:level2_timestimesoplus}}
\begin{center}
\begin{tabular}{|c|c|c|}
\hline 
$\piv_{001} $ & prob. & $\norm{\piv}^2$   
\\ \hline 
$\pv_3^3 = \pv_2^0 \oplus \pv_2^0 
 \eqs \left(1-\epsilon,\varepsilon^2,\dots,\varepsilon^2 \right)$ &
$(3p^2+2p+1)^2(1-p)^4 + \varepsilon$
 & $1- \varepsilon$ 
\\ \hline 
$\pv_3^4 = \left\{\begin{array}{c}
\pv_2^0 \oplus \pv_2^1 , \\
\pv_2^0 \oplus \pv_2^2
\end{array}
\right\} 
 \eqs  (\underbrace{C, \ldots, C}_{4 \text{-times}}, D, \dots, D)$& 
$2(3p^2 + 2p + 1)(4-3p)(1-p)^2p^3 $ & 
$4\left( \frac{1-p}{4-3p}\right)^2 $\\
with $C=\frac{1-p}{4-3p}+\varepsilon$ and $D=\frac{p}{(q-1)(4-3p)} + \varepsilon^2$ & $+ \varepsilon$ &
$+ \varepsilon$ \\
\hline
$\pv_3^5 = \left\{\begin{array}{c}
\pv_2^1 \oplus \pv_2^1 ,\\
\pv_2^1 \oplus \pv_2^2 ,\\
 \pv_2^2 \oplus \pv_2^2  
 \end{array}
\right\} 
 \eqs (\underbrace{K, \ldots, K}_{16 \text{-times}}, L, \dots, L )$ &
$(3p - 4)^2p^6 + \varepsilon$ & $16\left( \frac{1-p}{4-3p} \right)^4 + \varepsilon$\\
with $K = \left(\frac{1-p}{4-3p}\right)^2 + \varepsilon$ and $L = \frac{(8-7p)p}{(4-3p)^2(q-1)} + \varepsilon^2$ & & \\
\hline
\end{tabular}
\end{center}
\end{table}

\begin{lemma}
Let $\piv_{001}$ be the APP probability vector viewed by the decoder $U_{001}$ decoder.
For the channel error model of the code $U_{001}$ we have
\begin{eqnarray*}
\esp\left( \norm{\piv_{001}}^2 \right) & = & 
\frac{\left(5p^3-6p^2-5p-4\right)^2(1-p)^4}{(4-3p)^2} + \Oq
\end{eqnarray*}
\end{lemma}

\subsection{Computation of $\esp\left( \norm{\piv_{010}}^2\right)$}
It can be observed that $\piv_{010} = \piv \times \piv'$ where $\piv$ and $\piv'$ are independent and distributed as $\piv_{01}$ which is given in Table \ref{tab:level2_timesoplus}. Table \ref{tab:level2_timesoplustimes} summarizes this distribution.

\begin{table}[h!]
\caption{It is easy to check that $\norm{\pv_3^6}^2 = 1-\varepsilon$, 
$\norm{\pv_3^7}^2 = 4C^2 + \varepsilon$,
$\norm{\pv_3^8}^2 = 4M^2 + 4N^2 + \varepsilon$ and
$\norm{\pv_3^9}^2 = 8P^2 + \varepsilon$ with $C$, $M$, $N$ and $P$ as defined in this table.\label{tab:level2_timesoplustimes}}
\begin{center}
\begin{tabular}{|c|c|}
\hline 
$\piv_{010} $ & prob. 
\\
\hline 
$\pv_3^6 = \left\{\begin{array}{c}
\pv_2^3 \times \pv_2^3, \\
\pv_2^3 \times \pv_2^4, \\
\pv_2^3 \times \pv_2^5, \\
\pv_2^3 \times \pv_2^6, \\
\pv_2^3 \times \pv_2^7, \\
\pv_2^4 \times \pv_2^4, \\
\pv_2^4 \times \pv_2^6, \\
\pv_2^6 \times \pv_2^6
\end{array}\right\}
 = \left(1-\epsilon,\varepsilon^2,\dots,\varepsilon^2 \right)$ &
$(1-p)^4(1+4p+10p^2+4p^3-p^4) + \varepsilon$
\\
\hline 
$\pv_3^7 = \left\{\begin{array}{c}
\pv_2^4 \times \pv_2^5 ,\\
\pv_2^5 \times \pv_2^5 ,
 \end{array}
\right\} 
 \eqs (\underbrace{C, \ldots, C}_{4 \text{-times}}, D, \dots, D)$ &
$4p^3(1-p)^4(4-3p)+\varepsilon$ 
\\
with $C = \frac{1-p}{4-3p}+ \varepsilon$ and $D = \frac{p}{(4-3p)(q-1)} + \varepsilon^2$ & 
\\
\hline
$\pv_3^8 = \left\{\begin{array}{c}
\pv_2^4 \times \pv_2^7 ,\\
\pv_2^5 \times \pv_2^6 ,
\pv_2^6 \times \pv_2^7 ,
 \end{array}
\right\}  
 \eqs (M,M, \underbrace{N, \ldots, N}_{4 \text{-times}}, O, \dots, O)$ &
$4p^4(1-p)^2(7p^2-18p+12) + \varepsilon$ 
\\
with $M = \frac{(1-p)^2(4-3p)}{7p^2-18p+12} + \varepsilon$,
$N=\frac{(1-p)^2}{7p^2-18p+12} + \varepsilon$ & 
\\
and $O = \frac{p(4-3p)}{(7p^2-18p+12)(q-1)} + \varepsilon^2$ &  
\\
\hline
$\pv_3^9 = \left\{\begin{array}{c}
\pv_2^6 \times \pv_2^7 ,\\
\pv_2^7 \times \pv_2^7 ,
 \end{array}
\right\}  
\eqs(\underbrace{P, \ldots, P}_{8 \text{-times}}, Q, \dots, Q)$ &
$p^5(4-3p)(5p^2-12p+8)+\varepsilon$ 
\\
with $P = \frac{(1-p)^2}{5p^2-12p+8}+ \varepsilon$ and $Q = \frac{p(4-3p)}{(5p^2-12p+8)(q-1)} + \varepsilon^2$ &  
\\
\hline
\end{tabular}
\end{center}
\end{table}

\begin{lemma}
Let $\piv_{010}$ be the APP probability vector viewed by the decoder $U_{010}$ decoder.
For the channel error model of the code $U_{010}$ we have
\begin{eqnarray*}
\esp\left( \norm{\piv_{010}}^2 \right) & = & 
\frac{S(p)(1-p)^4}
{\left(7p^2 - 18p + 12\right)\left(5p^2 - 12p + 8\right)\left(3p - 4\right))} + \Oq
\end{eqnarray*}
with $S(p) = \left(151p^9 - 662p^8 + 1094p^7 - 1624p^6 + 4105p^5 - 6598p^4 + 4252p^3 - 272p^2 - 96p - 384\right)$
\end{lemma}

\subsection{Computation of $\esp\left( \norm{\piv_{011}}^2\right)$}
It can be observed that $\piv_{011} = \piv \oplus \piv'$ where $\piv$ and $\piv'$ are independent and distributed as $\piv_{01}$ which is given in Table \ref{tab:level2_timesoplus}. Table \ref{tab:level2_timesoplusoplus} summarizes this distribution.

\begin{table}[h!]
\caption{It is easy to check that $\norm{\pv_3^{10}}^2 = 1-\varepsilon$, 
$\norm{\pv_3^{11}}^2 = \norm{\pv_3^{12}}^2 = 2A^2 + \varepsilon$,
$\norm{\pv_3^{13}}^2 = 4G^2+ \varepsilon$,
$\norm{\pv_3^{14}}^2 = 8R^2+ \varepsilon$
and
$\norm{\pv_3^{15}}^2 = 16T^2 + \varepsilon$ with $A$, $G$, $R$ and $T$ as defined in this table.\label{tab:level2_timesoplusoplus}}
\begin{center}
\begin{tabular}{|c|c|}
\hline 
$\piv_{011} $ & prob.  
\\
\hline 
$\pv_3^{10} =
\pv_2^3 \oplus \pv_2^3 =  
 \left(1-\epsilon,\varepsilon^2,\dots,\varepsilon^2 \right)$ &
$(1-p)^8+\varepsilon$
\\
\hline 
$\pv_3^{11} = \pv_2^3 \oplus \pv_2^4  
 \eqs \left(A, A, B, \dots, B\right)$ & $8p(1-p)^7+\varepsilon$\\
with $A=\frac{1-p}{2-p} + \varepsilon$ and $B=\frac{p}{(2-p)(q-1)}+\varepsilon^2$ & \\
\hline
$\pv_3^{12} = 
\pv_2^3 \oplus \pv_2^5 \eqs \left(B, A, A, B, \dots, B\right)$ & $4p^2(1-p)^6+\varepsilon$ \\
with $A$ and $B$ as in the previous case & \\
\hline
$\pv_3^{13} = \left\{ \begin{array}{c}
\pv_2^3 \oplus \pv_2^6, \\
\pv_2^3 \oplus \pv_2^7, \\
\pv_2^4 \oplus \pv_2^4, \\
\pv_2^4 \oplus \pv_2^5, \\
\pv_2^5 \oplus \pv_2^5
\end{array}\right\} 
 \eqs (\underbrace{G, \ldots, G}_{4 \text{-times}}, H, \dots, H)$ &
$6(1-p)^4(2-p)^2p^2+\varepsilon$ 
\\
with $G = \left(\frac{1-p}{2-p}\right)^2+ \varepsilon$ and $H = \frac{p(4-3p)}{(2-p)^2(q-1)} + \varepsilon^2$ & 
\\
\hline
$\pv_3^{14} = \left\{ \begin{array}{c}
\pv_2^4 \oplus \pv_2^6, \\
\pv_2^4 \oplus \pv_2^7, \\
\pv_2^5 \oplus \pv_2^6, \\
\pv_2^5 \oplus \pv_2^7
\end{array}\right\} 
 \eqs (\underbrace{R, \ldots, R}_{8 \text{-times}}, S, \dots, S)$ & 
$4(1-p)^2(2-p)^3p^3+\varepsilon$ \\
with $R=\left(\frac{1-p}{2-p}\right)^3 + \varepsilon$ and $S=\frac{(7p^2-18p+12)p}{(2-p)^2(q-1)} + \varepsilon^2$ & \\
\hline
$\pv_3^{15} = \left\{\begin{array}{c}
\pv_2^6 \oplus \pv_2^6, \\
\pv_2^6 \oplus \pv_2^7, \\
\pv_2^7 \oplus \pv_2^7
\end{array}\right\} 
 \eqs (\underbrace{T, \ldots, T}_{16\text{-times}}, U, \dots, U)$ & 
$(2-p)^4p^4+\varepsilon$ \\
with $T=\left(\frac{1-p}{2-p}\right)^4 + \varepsilon$ and
$U=\frac{(5p^2-12p+8)p(4-3p)}{(2-p)^4(q-1)}+ \varepsilon^2$ & \\
\hline
\end{tabular}
\end{center}
\end{table}

\begin{lemma}
Let $\piv_{011}$ be the APP probability vector viewed by the decoder $U_{011}$ decoder.
For the channel error model of the code $U_{011}$ we have
\begin{eqnarray*}
\esp\left( \norm{\piv_{010}}^2 \right) & = & 
\frac{(p+2)^4(p-1)^8}{(2-p)^4}+ \Oq
\end{eqnarray*}
\end{lemma}

\subsection{Computation of $\esp\left( \norm{\piv_{100}}^2\right)$,
$\esp\left( \norm{\piv_{101}}^2\right)$,
$\esp\left( \norm{\piv_{110}}^2\right)$ and 
$\esp\left( \norm{\piv_{111}}^2\right)$}

Note that $\piv_{110}$ and $\piv_{111}$ are distributed like $\piv \times \piv'$ and $\piv \oplus \piv'$ where $\piv$ and $\piv'$ are independent and distributed like $\piv_{11}$ which is itself the APP vector obtained from transmitting $0$ over a $q\hbox{-SC}_{p_2+\varepsilon}$ with $p_2=2p_1-p_1^2$. We can therefore use directly Lemmas \ref{L12} and \ref{L11} and obtain

\begin{lemma}
Let $\piv_{110}$ and $\piv_{111}$ be the APP probability vectors viewed by the decoder $U_{110}$ and $U_{111}$, respectively. For the channel error model of the code $U_{110}$ we have
\begin{eqnarray*}
\esp\left( \norm{\piv_{110}}^2 \right) & = & \frac{(2+p_2)(1-p_2)^2}{(2-p_2)} + \Oq\\
& = & \frac{-\left(p^4 - 4p^3 + 6p^2 - 4p - 2\right)(p - 1)^8}{(p^4 - 4p^3 + 6p^2 - 4p +
2)} + \Oq
\end{eqnarray*}
The channel error model for the code $U_{111}$ is a $q\hbox{-SC}_{p_3+\varepsilon}$ with $p_3=2p_2-p_2^2$
\begin{eqnarray*}
\esp\left( \norm{\piv_{111}}^2 \right) &= & (1-p_3)^2 + \Oq = (1-p)^{16} + \Oq 
\end{eqnarray*}
\end{lemma}

Note that $\piv_{100}$ and $\piv_{101}$ are distributed like $\piv \times \piv'$
and $\piv \oplus \piv'$ where $\piv$ and $\piv'$ are independent and distributed like $\piv_{10}$.
We can therefore use directly Lemmas \ref{L21} and \ref{L22} and obtain

\begin{lemma}
Let $\piv_{100}$ and $\piv_{101}$ be the APP probability vectors viewed by the decoder $U_{100}$ and $U_{101}$, respectively. For the channel error model of the code $U_{100}$ and $U_{101}$ we have
\begin{eqnarray*}
\esp\left( \norm{\piv_{100}}^2 \right) &= & \frac{(5p_1^3-6p_1-5p_1-4)(1-p_1)^2}{(3p_1-4)}+ \Oq \\
& = &
\frac{\left(5p^6 - 30p^5 + 66p^4 - 64p^3 + 19p^2 + 10p + 4\right)(p - 1)^4}{(3p^2- 6p + 4)} + \Oq \\
\esp\left( \norm{\piv_{101}}^2 \right) &= & \frac{(p_1+2)^2(1-p_1)^4}{(2-p_1)^2} + \Oq = 
\frac{(p^2 - 2p - 2)^2(p - 1)^8}{(p^2 - 2p + 2)^2} + \Oq
\end{eqnarray*}

\end{lemma}

\section{The Koetter-Vardy decoding algorithm for AG codes}
\label{sec:KV_AG}

The Koetter-Vardy decoding algorithm for Reed-Solomon codes can be adapted to AG codes as was shown in \cite{KV03a}. 
In this appendix, we give a short description of this algorithm and a review of the main results of \cite{KV03a} that  we need
to  prove
Theorem \ref{th:listsize_AG}. This section is essentially nothing but a subset of results presented for AG codes in the preprint version 
\cite{KV03a} which we repeat here for the convenience of the reader since the additional material present in the preprint version
has not been published as far as we know.

We first begin with the notion of a {\em gap} which will be useful to describe the Koetter-Vardy soft decoding algorithm 
for AG codes. We consider an algebraic curve $\X$ defined over a finite field $\Fq$ of genus $g$. 
Let $Q$ be a rational point on $\X$. We also assume that $\X$ has at least $n$ other rational points $P_1,\dots,P_n$ besides $\X$.
 A positive integer $i$ is called a {\em (Weierstrass) gap} at $Q$ if $\mathcal L(iQ)=\mathcal L((i-1)Q)$. Otherwise $i$ is a {\em non-gap} at $Q$. It is well known that gaps at $Q$ lie in the interval $[0, 2g-1]$ and that the number of gaps is equal to $g$.

These gaps at $Q$ can be used to construct a basis for the space $\mathcal L(mQ)$: we fix an arbitrary rational function $\phi_i \in \mathcal L(iQ) \setminus \mathcal L((i-1)Q)$ if $i$ is a non-gap at $Q$ and we set $\phi_i = 0$ otherwise, for $i\in \{0, \ldots, m\}$.

The ring of rational functions that have either no pole or just one pole at $Q$ which is defined by 
$$\mathcal K_Q \eqdef \bigcup_{i=0}^{\infty} \mathcal L(iQ)$$
will also be helpful in what follows.
In other words, 
we can write any polynomial $A(Y) \in \mathcal K_Q[Y]$ in a unique way as $A(Y)=\sum_{i,j} a_{i,j}\phi_i Y^j$.
This allows to define for a pair $(w_Q, w_y)$ of nonnegative real numbers the 
\emph{$(w_Q, w_y)$-weighted $Q$-valuation} of $A(Y)$, denoted by $\deg_{w_Q,w_y}(A(Y))$, which is  the maximum over all numbers $iw_Q + jw_y$ such that $a_{i,j}\neq 0$.

We will also need the notion of the multiplicity of a polynomial in $\mathcal K_Q[Y]$ at a certain point $P$.
For this purpose, it will be convenient to introduce a new basis for $\mathcal L(mQ)$. 
We define $\phi_{0,P},\phi_{1,P},\dots,\phi_{m,P}$ as follows. If there exists at least one $f \in \mathcal L(mQ)$ that has multiplicity
exactly $i$ at $P$ we set $\phi_{i,P}$ to be one of these functions (we make an arbitrary choice if there are several functions of this kind).
If there is no such function we set $\phi_{i,P}= 0$. For the case we are interested in, namely 
$n > m \geq 2g-1$, it is known that there are exactly $g$ indices $i$ for 
which $\phi_{i,P}= 0$. It is known \cite{S93a} that 
the set of functions among $\phi_{0,P},\phi_{1,P},\dots,\phi_{m,P}$ which are non zero form a basis of $\mathcal K_Q[Y]$.
We write from now on each $f$ in $\mathcal L(mQ)$ in a unique way as
$$
f = \sum_{i=0}^m a_{i} \phi_{i,P}
$$
when we assume that $a_{i}=0$ if $\phi_{i,P}=0$.
\begin{definition}[multiplicity of a polynomial in {$\mathcal K_Q[Y]$}]
Let $A(Y)$ be a polynomial in $\mathcal K_Q[Y]$ and consider the shifted polynomial $A(Y+\alpha)$ expressed using the above basis, that is,
\begin{equation}
\label{equation-multiplicity}
A(Y+\alpha) = \sum_{i,j} b_{i,j} \phi_{i,P} (Y+\alpha)^j
\end{equation}
we say that $A(Y)$ has a \emph{zero of multiplicity $m$} at the interpolation point $(P,\alpha)$ if $b_{i,j} = 0$ for $i+j<m$ and there exists a nonzero coefficient $b_{i,j}$ with $i+j=m$.
\end{definition}

We are ready now for describing the Koetter-Vardy decoding algorithm for AG codes.
We consider here an  AG code $\C{\X}{\P}{mQ}$ of length $n$ over $\Fq$ defined by a set of $n+1$ distinct $\mathbb F_q$-rational points:  $Q$ and $\mathcal P =\left\{P_1, \ldots, P_n\right\}$. We are also given a multiplicity matrix $\Mm=\left(m_{\alpha,j}\right)_{\substack{\alpha \in \Fq \\
1 \leq j \leq n}}$. 

\noindent
{\bf Interpolation step:} It  consists in computing the (nontrivial) polynomial $Q_{\Mm}(Y)\in \mathcal K_Q[Y]$ of minimal $(1,m)$-weighted $Q$-valuation that has a zero of multiplicity at least $m_{\alpha,j}$ at the interpolation point $(P_j,\alpha)$.

\noindent
{\bf Factorization step:}  It consists in identifying all the factors of $Q_{\Mm}(Y)$ of type $Y-f$ with $f\in \mathcal L(mQ)$. 
The output of the algorithm is a list of the codewords that correspond to these factors.

The following quantities will be useful for understanding this algorithm.
\begin{itemize}
\item The number of monomials whose $(w_Q, w_y)$-weighted $Q$-valuation is at most $\delta$ is denoted $N_{w_Q, w_y} (\delta, \X)$. Thus:
$$N_{w_Q, w_y} (\delta, \X) \eqdef |\left\{ (i,j)\mid i,j \geq 0,\text{ $\phi_i \neq 0$ and $iw_Q + jw_y \leq \delta$}\right\}|$$
\item We define the inverse function
$$\Delta_{w_Q, w_y} (\nu, \X) \eqdef \min \left\{ \delta \in \mathbb Z \mid N_{w_Q, w_y} (\delta, \X) > \nu \right\} $$
\end{itemize}

To get a better understanding of  the soft-decision algorithm of AG codes the following theorem will be very helpful 
\begin{theorem}[Theorem 18 and Corollary 20 \cite{KV03a}]
\label{AG-1}
Let $C=C(\Mm)$ denote the cost of the multiplicity matrix $\Mm$.
The list obtained by factoring the interpolation polynomial $Q_{\Mm}(Y)$ contains a codeword $\mathbf c \in \C{\X}{\P}{mQ}$ if
\begin{equation}
\label{eq:SC} 
\left\langle \Mm, \left\lfloor \mathbf c\right\rfloor\right\rangle >  \Delta_{1,m}(C,\mathcal X) 
\end{equation}
We have the following upper bound on $\Delta_{1,m}(C,\mathcal X)$:
\begin{equation}
\label{eq:upper_bound_Delta}
\Delta_{1,m}(C,\mathcal X) \leq g + \sqrt{2m(C+g)+g^2}.
\end{equation} 
\end{theorem}

\begin{proof}

We first prove that if the condition \eqref{eq:SC} holds then the list contains the codeword $\cv$.

Recall that for every $\mathbf c\in \C{\X}{\P}{mQ}$ there exists a rational function $f \in \mathcal L(mQ)$ such that $c_j = f(P_j)$ for $j=1, \ldots, n$. 
Given the interpolation polynomial $Q_{\Mm}(Y)\in \mathcal K_{Q}[Y]$,
we consider the function $h\in \mathcal K_Q$ defined by $h=Q_{\Mm}(f)$.
By construction, $Q_{\Mm}(Y)$ passes trough the points $(P_\ell,c_\ell)$ with multiplicities at least $m_\ell$ where $\left\langle \Mm,\lfloor \mathbf c \rfloor\right\rangle = m_1 + \cdots + m_n$. Then:
\begin{itemize}
\item 
We claim that  the function $h=Q_{\Mm}(f)$ has at least $\left\langle \Mm,\lfloor \mathbf c \rfloor\right\rangle$ zeros in $\mathcal P$ counted with multiplicities. 
Indeed, if $Q_{\Mm}(Y)$ passes through the interpolation point $(P_\ell, c_\ell)$ with multiplicity at least $m_\ell$ and we express $Q_{\Mm}(Y)$  in the basis 
of the $\Phi_{i,P_\ell}$'s we have
$$Q_{\Mm}(Y) = \sum_{i,j} a_{i,j} \phi_{i,P_\ell} (Y-c_\ell)^j.$$
But, by \eqref{equation-multiplicity}
it is required that $a_{i,j}=0$ if $i+j<m_\ell$. We thus get that $h=Q_{\Mm}(f) = \sum_{i,j} a_{i,j} \phi_{i,P_\ell} (f-c_ell)^j$ has a zero of multiplicity $m_\ell$ at the point $P_\ell$
since $f(P_\ell)=c_\ell$.

\item Since $f\in \mathcal L(mQ)$ and $\deg_{1,m}(Q_{\Mm}(Y))\leq \Delta_{1,m}(C,\mathcal X)$ then $h$ has at most $\Delta_{1,m}(C,\mathcal X)$ poles at $Q$. And these are its only poles since $h\in \mathcal K_Q$.

\end{itemize}
That is, if $\left\langle \Mm,\lfloor \mathbf c \rfloor\right\rangle > \Delta_{1,m}(C,\mathcal X)$, then $h$ has more zeros than poles. Thus, $Q_{\Mm}(f) = h \equiv 0$, in other words, $Y-f$ is a factor of $Q_{\Mm}(Y)$.

For the second statement of the theorem, let  $A(X,Y)=\sum_{i=1}^{\infty} \sum_{j=1}^{\infty} a_{i,j}X^i Y^j$ be a bivariate polynomial over $\mathbb F_q$ and let $w_X, w_Y\in \mathbb R$. In \cite{KV03}, the $(w_X, w_Y)$-weighted degree of $A(X,Y)$ is defined as the maximum over all numbers $iw_X + jw_Y$ such that $a_{i,j}\neq 0$. Moreover, the number of monomials of $(w_X,w_Y)$-weighted degree at most $\delta$ is denoted in \cite{KV03} as $N_{w_X,w_Y}(\delta)$. That is,
$$N_{w_X, w_Y}(\delta) \eqdef \left| \left\{ X^i Y^j \mid i,j \geq 0 \hbox{ and } iw_X + jw_Y \leq \delta\right\}\right|$$

It is easy to see that
$$N_{w_Q,w_Y}(\delta, \mathcal X) \geq N_{w_Q, w_Y}(\delta) - g\left\lfloor \frac{\delta}{w_Q} + 1\right\rfloor$$
Indeed, the number of different expressions $\phi_iY^j$ in $N_{w_Q,w_Y}(\delta, \mathcal X)$ is equal to $N_{w_Q, w_Y}(\delta)$ but taking into account that some of the functions $\phi_i$ are zero. Then, the result follows from the fact that the number of functions $\phi_i$ that are zero, or equivalent, the number of gaps, is bounded by the genus $g$ of the curve $\mathcal X$; and the fact that the number of monomials $Y^j$ such that $jw_Q \leq \delta$ is upper bounded by $\frac{\delta}{w_Q}$.

Thus, using \cite[Lemma1]{KV03} we have
$$N_{1,m}(\delta, \mathcal X) \geq N_{1,m}(\delta) - g\left\lfloor \frac{\delta}{m}+1\right\rfloor > \frac{\delta^2}{2m} - g\left( \frac{\delta}{m}
+1\right)$$

By replacing $N_{1,m}(\delta, \mathcal X)$ by $\nu$ we can write the above expression as:
$(\delta - g)^2 < 2m(\nu + g) + g^2$.
Then, by the definition of $\Delta_{1,m}(\nu, \mathcal X)$, we get the following upper bound:
$$\Delta_{1,m}(\nu, \mathcal X) \leq g + \sqrt{2m(\nu + g) + g^2}$$

\end{proof}

As for Reed-Solomon codes we can obtain an \emph{algebraic soft-decoding for AG codes with list size limited to $L$}. In 
the following we adapt the  ideas of \cite{KV03} to AG codes.

\begin{lemma}
\label{AG-2}
The number of codewords on the list produced by the soft-decision decoder for the AG code $\C{\X}{\P}{mQ}$ with a given multiplicity matrix $\Mm$ does not exceed
$$L_m(\Mm) \eqdef \frac{g+\sqrt{2m(C+g)+g^2}}{m}$$
where $g$ denotes the genus of the curve $\X$ and $C=C(\Mm)$ the cost of matrix $\Mm$.
\end{lemma}

\begin{proof}
Similar to \cite[Lemma 15]{KV03} the number of codewords of the list is upper-bounded by the $(0,1)$-weight $Q$-valuation of the interpolation polynomial $Q_{\Mm}(Y)$. By definition of weighted $Q$-valuation, we have:
$$\deg_{0,1} Q_{\Mm}(Y)\leq \frac{\deg_{1,m} Q_{\Mm}(Y)}{m}\leq \frac{\Delta_{1,m}(C, \mathcal X)}{m}\leq \frac{g+\sqrt{2m(C+g)+g^2}}{m} $$
where the first inequality follows from the definition of weighted $Q$-valuation, the second inequality follows from the definition of $\Delta_{1,m}(C, \X)$ and the third inequality follows from Theorem \ref{AG-1}.
\end{proof}

\begin{remark}
\label{remark:AG1}
The very definition of $L_m(\Mm)$ implies
\begin{equation}
\label{eq:CM_LM}
2C(\Mm) = m L_m(\Mm)^2 - 2L_m(\Mm)g - 2g
\end{equation}
\end{remark}

\begin{lemma}
\label{AG-3}
For a given multiplicity matrix $\Mm$, the algebraic soft-decision decoding algorithm outputs a list that contains a codeword $\mathbf c \in \C{\X}{\P}{mQ}$ if
$$\frac{\left\langle \Mm,\left\lfloor \mathbf c\right\rfloor\right\rangle}{\sqrt{\left\langle \Mm,\Mm\right\rangle + \left\langle \Mm,1\right\rangle}}\geq \sqrt{m} + \frac{2g+\sqrt{2mg}}{\sqrt{2C(\Mm)}}$$
where $g$ denotes the genus of the curve $\X$ and $C=C(\Mm)$ the cost of matrix $\Mm$.
\end{lemma}

\begin{proof}
This follows immediately from Theorem \ref{AG-1}
and the fact that 
$$\sqrt{2mC} + \sqrt{2mg} + 2g \geq g + \sqrt{2m(C+g) + g^2} \hbox{ for all } g,m,C\geq 0$$
\end{proof}

Let $\Pim$ be a given reliability matrix and $\Mm$ be the corresponding multiplicity matrix produced by \cite{KV03}[Algorithm A]. Then, by \cite{KV03}[Lemma 16] there exists a positive real number $\lambda$ such that 
$\Mm=\lambda \Pim - \Jm$
where $\Jm$ denotes a $q\times n$ matrix whose entries are all nonnegative real numbers not exceeding $1$.
Then:
$$\left\langle \Mm, \Mm\right\rangle + \left\langle \Mm,\one\right\rangle = \lambda^2 \left\langle
\Pim, \Pim \right\rangle - \lambda \left\langle \Pim, 2\Jm-\one\right\rangle + \left\langle \Jm, \Jm-\one\right\rangle$$
By definition we have that
\begin{equation}
\label{Equation-List}
\left\langle \Mm, \Mm\right\rangle + \left\langle \Mm,\one\right\rangle = 2C(\Mm)
\end{equation}

We can use both expressions to get a quadratic equation in $\lambda$ which has only one positive root:

\begin{equation}
\label{Eq-AG-2}
\lambda = \underbrace{\frac{\left\langle \Pim, 2\Jm-\one\right\rangle}{2\left\langle \Pim, \Pim\right\rangle}}_{\lambda_2} + \underbrace{\sqrt{\frac{\left\langle \Pim, 2\Jm-\one\right\rangle^2}{4\left\langle \Pim, \Pim\right\rangle^2} + \frac{\left\langle \Jm, \one-\Jm\right\rangle}{\left\langle \Pim, \Pim\right\rangle} + \frac{2C(\Mm)}{\left\langle \Pim, \Pim\right\rangle}}}_{\lambda_1}
\end{equation}

\subsection{Proof of Theorem \ref{th:listsize_AG}}

By using $\Mm=\lambda\Pim - \Jm$, we can reformulate the sufficient condition of Theorem \ref{AG-3}, 
by rewriting $\left\langle \Mm,\lfloor \cv \rfloor\right\rangle = \left\langle \lambda\Pim - \Jm,\lfloor \cv \rfloor\right\rangle
= \lambda  \left\langle \Pim  ,\lfloor \cv \rfloor\right\rangle -  \left\langle  \Jm,\lfloor \cv \rfloor\right\rangle$ and obtain
$$
\frac{\lambda  \left\langle \Pim  ,\lfloor \cv \rfloor\right\rangle -  \left\langle  \Jm,\lfloor \cv \rfloor\right\rangle}{\sqrt{2C(\Mm)}}
\geq \sqrt{m} + \frac{g+\sqrt{2mg}}{\sqrt{2C(\Mm)}}
$$
which is equivalent to
$$\frac{\left\langle \Pim, \lfloor \mathbf c\rfloor\right\rangle}{\sqrt{\left\langle \Pim, \Pim\right\rangle}}
\left( \lambda - \frac{\left\langle \Jm,\lfloor \mathbf c\rfloor\right\rangle}{\left\langle \Pim, \lfloor \mathbf c \rfloor \right\rangle}\right) 
\frac{\sqrt{\left\langle \Pim, \Pim\right\rangle}}{\sqrt{2C(\Mm)}}
\geq \sqrt{m} + \frac{g+\sqrt{2mg}}{\sqrt{2C(\Mm)}}$$

Using the expression for $\lambda$ in Equation \eqref{Eq-AG-2}, we can express the previous formula as:
\begin{equation}
\label{eq:F1F2F3}
\frac{\left\langle \Pim, \lfloor \mathbf c\rfloor\right\rangle}{\sqrt{\left\langle \Pim, \Pim\right\rangle}}
(F_1 - F_2 - F_3)\geq \sqrt{m} + \frac{g+\sqrt{2mg}}{\sqrt{2C(\Mm)}} 
\end{equation}
with:
\begin{eqnarray*}
F_1 &\eqdef &\lambda_1 \frac{\sqrt{\left\langle \Pim, \Pim\right\rangle}}{\sqrt{2C(\Mm)}} \\
&=  &\sqrt{\frac{\left\langle \Pim, 2\Jm-\one\right\rangle^2}{4\left\langle \Pim, \Pim\right\rangle^2} + \frac{\left\langle \Jm, \one-\Jm\right\rangle}{\left\langle \Pim, \Pim\right\rangle} + \frac{2C(\Mm)}{\left\langle \Pim, \Pim\right\rangle}}\frac{\sqrt{\left\langle \Pim, \Pim\right\rangle}}{\sqrt{2C(\Mm)}}\\
& = & \sqrt{1 + \frac{\left\langle \Jm, \one-\Jm\right\rangle}{2C(\Mm)}+ \frac{\left\langle \Pim, 2\Jm-\one\right\rangle^2}{8C(\Mm)\left\langle \Pim, \Pim\right\rangle}}\\
& \geq &1
\end{eqnarray*}
$$
F_2 \eqdef  -\lambda_2 \frac{\sqrt{\left\langle \Pim, \Pim\right\rangle}}{\sqrt{2C(\Mm)}} = \frac{1}{\sqrt{2C(\Mm)}}\frac{\left\langle \Pim, \one-2\Jm\right\rangle}{2\sqrt{\left\langle \Pim, \Pim\right\rangle}}
\leq \frac{1}{\sqrt{2C(\Mm)}} \frac{n}{2\sqrt{\tfrac{n}{q}}}$$
To obtain the previous inequality we have used the fact that $\left\langle \Pim, \one-2\Jm\right\rangle \leq \left\langle \Pim, \one\right\rangle = n$ and $\left\langle \Pim, \Pim\right\rangle \geq \frac{n}{q}$.
$$F_3 \eqdef \frac{\left\langle \Jm,\lfloor \mathbf c\rfloor\right\rangle}{\left\langle \Pim, \lfloor \mathbf c \rfloor \right\rangle} \frac{\sqrt{\left\langle \Pim, \Pim\right\rangle}}{\sqrt{2C(\Mm)}} 
\leq \frac{1}{\sqrt{2C(\Mm)}} \frac{n}{\sqrt{m} + \frac{g+\sqrt{2mg}}{\sqrt{2C(\Mm)}}}
= \frac{n}{\sqrt{2C(\Mm)m} + g + \sqrt{2mg}} \leq \frac{1}{\sqrt{2C(\Mm)}} \frac{n}{\sqrt{m}}$$
To obtain the previous inequality we have made use of the following two observations:
\begin{itemize}
\item  $\left\langle \Jm, \lfloor \mathbf c\rfloor\right\rangle \leq \left\langle \one, \lfloor \mathbf c \rfloor\right\rangle = n$;
\item if $\Pim$ and $\mathbf c$ are such that \eqref{eq:listsize_AG} holds, then {\em a fortiori}:
$$\frac{\left\langle \Pim, \lfloor \mathbf c\rfloor \right\rangle}{\sqrt{\left\langle \Pim, \Pim\right\rangle}} \geq \sqrt{m} + \frac{2g+\sqrt{2mg}}{\sqrt{2C(\Mm)}}$$
\end{itemize}

Using all the bounds that we have just given for the $F_i$'s we obtain that \eqref{eq:F1F2F3} holds when
$$
\frac{\left\langle \Pim, \lfloor \mathbf c\rfloor\right\rangle}{\sqrt{\left\langle \Pim, \Pim\right\rangle}}
\left( 1 - \frac{1}{\sqrt{2C(\Mm)}} \frac{n}{2\sqrt{\tfrac{n}{q}}} -  \frac{1}{\sqrt{2C(\Mm)}} \frac{n}{\sqrt{m}}\right)\geq \sqrt{m} + \frac{g+\sqrt{2mg}}{\sqrt{2C(\Mm)}} 
$$
which is equivalent to
\begin{equation}
\label{eq:new}
\frac{\left\langle \Pim, \lfloor \mathbf c\rfloor\right\rangle}{\sqrt{\left\langle \Pim, \Pim\right\rangle}}
\geq 
\frac{\sqrt{m} + \frac{g+\sqrt{2mg}}{\sqrt{2C(\Mm)}}}{1 - \frac{1}{\sqrt{2C(\Mm)}} \frac{n}{2\sqrt{\tfrac{n}{q}}} -  \frac{1}{\sqrt{2C(\Mm)}} \frac{n}{\sqrt{m}}}
\end{equation}

From $L\leq L_m(\Mm) \leq L+1$ and $2C(\Mm)= mL_m^2(\Mm) - 2L_m(\Mm)g - 2g$ we deduce that
\begin{equation}
\label{eq:lower_bound_C}
\sqrt{2C(\Mm)} \geq \sqrt{mL^2- 2(L+1)g - 2g} = L\sqrt{m} \sqrt{1-\tfrac{2g}{mL}\left( 1+\tfrac{2}{L}\right)}
\end{equation}

Therefore \eqref{eq:new} holds if we have
$$
\frac{\left\langle \Pim, \lfloor \mathbf c\rfloor\right\rangle}{\sqrt{\left\langle \Pim, \Pim\right\rangle}}
\geq 
\frac{\sqrt{m} + \frac{g+\sqrt{2mg}}{L\sqrt{m} \sqrt{1-\tfrac{2g}{mL}\left( 1+\tfrac{2}{L}\right)}}}{1 - \frac{1}{L\sqrt{m} \sqrt{1-\tfrac{2g}{mL}\left( 1+\tfrac{2}{L}\right)}}
\left( \frac{n}{2\sqrt{\tfrac{n}{q}}} +   \frac{n}{\sqrt{m}} \right)}
$$

By introducing $\tilde{R}= \frac{m}{n}$ and $\tilde{\gamma} \eqdef \frac{g}{m}$ we obtain that 
this inequality can be rewritten as
$$
\frac{\left\langle \Pim, \lfloor \mathbf c\rfloor\right\rangle}{\sqrt{\left\langle \Pim, \Pim\right\rangle}}
\geq \sqrt{m} 
\frac{1+ \frac{\tilde{\gamma}+\sqrt{2 \tilde{\gamma}}}{L \sqrt{1-\tfrac{2\tilde{\gamma}}{L}\left( 1+\tfrac{2}{L}\right)}}}{1 - \frac{1}{L \sqrt{1-\tfrac{2 \tilde{\gamma}}{L}\left( 1+\tfrac{2}{L}\right)}}
\left( \frac{\sqrt{q }}{2 \sqrt{\tilde{R}}} +   \frac{1}{\tilde{R}} \right)} = \sqrt{m}\left(1 + \OO{\frac{1}{L}} \right)
$$

This completes the proof of the theorem, that is, $\mathbf c\in \C{\X}{\P}{mQ}$ is on the list produced by the soft-decision decoder, provided that
this last inequality is satisfied.

\end{document}